\documentclass[a4paper]{article}

\usepackage{microtype}%
\usepackage{xcolor}
\usepackage{xspace}
\usepackage{algorithm}
\usepackage[noend]{algpseudocode}
\usepackage{xifthen}
\usepackage{tikz}
\usetikzlibrary{arrows}
\usepackage{dsfont}

\usepackage[top=3cm, bottom=3.5cm, left=4cm, right=4cm]{geometry}

\usepackage{amsmath} %
\usepackage{amsthm} %
\usepackage{listings}
\usepackage{subcaption}

\algrenewcommand{\algorithmiccomment}[1]{$\rhd$ #1}
\algrenewcommand\algorithmicprocedure{\textbf{action}}

\newtheorem{definition}{Definition}
\newtheorem{theorem}{Theorem}
\newtheorem{lemma}{Lemma}
\newtheorem{corollary}{Corollary}

\newcommand*{\N}{\mathds{N}}

\newcommand{\RID}{\textsc{RID}\xspace}
\newcommand{\notauthorized}[1]{\textsc{not\_authorized(\ensuremath{#1})}\xspace}
\newcommand{\probe}[1]{\textsc{probe(\ensuremath{#1})}\xspace}
\newcommand{\probefail}[1]{\textsc{probefail(\ensuremath{#1})}\xspace}

\newcommand{\outrelayclosed}[1]{\textsc{out-relay-closed(\ensuremath{#1})}\xspace}
\newcommand{\inrelayclosed}[1]{\textsc{in-relay-closed(\ensuremath{#1})}\xspace}
\newcommand{\ping}[1]{\textsc{ping(\ensuremath{#1})}\xspace}
\newcommand{\transmit}[1]{\textsc{transmit(\ensuremath{#1})}\xspace}

\newcommand{\timeout}{\textsc{Timeout}\xspace}

\newcommand{\send}[1]{\textsc{send(\ensuremath{#1})}\xspace}

\newcommand{\myref}[1]{\ensuremath{\hat{#1}}\xspace}

\newcommand{\relayprims}{\ensuremath{\mathcal{IFR}}\xspace}

\newcommand{\incoming}{\ensuremath{\operatorname{incoming}}\xspace}
\newcommand{\direct}{\ensuremath{\operatorname{direct}}\xspace}
\newcommand{\sameTarget}{\ensuremath{\operatorname{same-target}}\xspace}

\makeatletter
\newenvironment{proofof}[1]{\par
  \normalfont \topsep6\p@\@plus6\p@\relax
  \trivlist
  \item[\hskip\labelsep
        \bfseries
    Proof of #1\@addpunct{.}]\ignorespaces
}{%
  \qed
  \endtrivlist\@endpefalse
}
\makeatother

\date{}

\usetikzlibrary{arrows}
\usetikzlibrary{fit}
\usetikzlibrary{shadows}
\usetikzlibrary{patterns}
\usetikzlibrary{matrix}
\usetikzlibrary{shapes}
\usetikzlibrary{calc}
\usetikzlibrary{positioning} 
\usetikzlibrary{backgrounds}
\usetikzlibrary{decorations.markings}
\usetikzlibrary{decorations.pathreplacing}

\begin{document}

 \title{Relays: A New Approach for the Finite Departure Problem in Overlay Networks\thanks{This work was partially supported by the German Research Foundation (DFG) within the Collaborative Research Center ``On-The-Fly Computing'' (SFB 901).}}

\author{Christian Scheideler \and Alexander Setzer}

\maketitle

\begin{abstract}
A fundamental problem for overlay networks is to \emph{safely} exclude leaving
nodes, i.e., the nodes requesting to leave the overlay network are excluded
from it without affecting its connectivity. To rigorously study
self-stabilizing solutions to this problem, the {\em Finite Departure Problem}
(FDP) has been proposed \cite{DBLP:conf/sss/ForebackKNSS14}. In the FDP we are
given a network of processes in an arbitrary state, and the goal is to
eventually arrive at (and stay in) a state in which all leaving processes
irrevocably decided to leave the system while for all weakly-connected
components in the initial overlay network, all staying processes in that
component will still form a weakly connected component. In the standard
interconnection model, the FDP is known to be unsolvable by local control
protocols, so oracles have been investigated that allow the problem to be
solved \cite{DBLP:conf/sss/ForebackKNSS14}. To avoid the use of oracles, we
introduce a new interconnection model based on relays. Despite the relay model
appearing to be rather restrictive, we show that it is universal, i.e., it is
possible to transform any weakly-connected topology into any other
weakly-connected topology, which is important for being a useful
interconnection model for overlay networks. Apart from this, our model allows
processes to grant and revoke access rights, which is why we believe it to be
of interest beyond the scope of this paper. We show how to implement the relay
layer in a self-stabilizing way and identify properties protocols need to
satisfy so that the relay layer can recover while serving protocol requests.
\end{abstract}

\section{Introduction}

Once distributed systems become large enough, membership changes in these
systems are not an exception but the norm. This particularly holds for
peer-to-peer systems but is also true for large server-based systems as
servers may need to be taken offline for some maintenance or new servers need
to be included in the system to improve or maintain the service quality. So
protocols need to be in place to allow members of a distributed system to join
and leave it without disrupting its functionality. The most basic requirement
for maintaining the functionality of a system is that it stays weakly
connected. While this is easy to guarantee when new members join a system, it
is not so easy to guarantee when members leave the system, in particular, if
multiple members want to leave the system at the same time. In the literature
on peer-to-peer systems, many proposals for leave protocols have already been
made
(see, e.g., \cite{SaiaT08,KSW10,HayesST12,PRT14,APRRU15,DGS16}). However, most
of these solutions cannot give any guarantees if the system is not in some
well-defined state. Distributed systems can easily be pushed into a
non-well-defined state if there are network partitions or faulty members, so
it would be desirable to have leave protocols that do not need any assumptions
on the system state.

In order to rigorously study the problem of guaranteeing weak connectivity for
any situation in which a collection of members (henceforth also simply called
{\em processes}) wants to leave the system, Foreback et
al.~\cite{DBLP:conf/sss/ForebackKNSS14} introduced the \emph{Finite Departure
Problem (FDP)}. In the FDP the leaving processes have to irrevocably decide in
finite time when it is safe to leave the network, i.e., their departure does
not cause the network to get disconnected. Foreback et al. showed that there
is no self-stabilizing local-control protocol for the FDP. At the heart of the
proof are two serious problems: The standard assumption used in overlay
networks research that a process may freely pass knowledge about its neighbors
to any one of its neighbors has the effect that a process $v$ cannot locally
decide whether $v$ is critical for the connectivity of the network or not,
simply because it does not have any control on and thereby potentially
incomplete knowledge about its incoming connections (i.e., the set of
processes knowing its address). Also, when assuming asynchronous
communication, where message may have arbitrary finite delays, a process $v$
may not know whether messages carrying critical connectivity information are
still on their way to $v$. This caused Foreback et al. to introduce the NIDEC
oracle, which gives a process $v$ the power to determine whether its address
is still known somewhere in the system (NID is a short form of "no ID") and
whether there are still messages on their way to $v$ (EC is a short form of
"empty channel").

Is it possible to avoid the use of oracles by using a different link layer
model? We show that this is indeed the case. In fact, we need two layers: a
self-stabilizing link layer and, on top of that, a self-stabilizing relay
layer, which is our main innovation. On top of the relay layer, a
self-stabilizing local-control protocol can then be designed to solve the FDP
problem without the use of an oracle. While the link layer ensures that a
process is aware of the messages that are still in transit along its outgoing
connections, the relay layer gives the processes the power to rigorously
control who is allowed to send messages to it. Despite appearing to be rather
restrictive, we show that the relay concept is universal in a sense that one
can get from any weakly connected topology to any other weakly connected
topology while staying weakly connected throughout the transformation process.
Because the relay layer now allows the rigorous study of access control
problems in overlay networks, which opens up new directions like rigorous
studies on the DoS-resistance of overlay networks (given that messages can
only be sent via relay connections), we expect it to be of interest beyond
this paper.

\subsection{System model}

We consider a distributed system consisting of a set of processes that are interconnected to each other (with more details on the type of interconnections once we introduce relays in the next section). The processes are controlled by a local-control protocol that specifies the variables and actions that are available in each process.
We assume that there is a reliable link layer that transmits messages from processes to other processes based on an ID of the target process contained in the message.
More specifically, each process specifies a set of variables, called \emph{buffers} containing messages to be sent to other processes and the ID of the respective target process is stored along with the buffer or inside the message.
We assume that the link layer may take an arbitrary but finite amount of time to process a message that was put into one of these variables, but messages never get lost.
We assume the link layer makes sure that every transmitted message will eventually be removed from the buffer it was taken from after it has been processed by the receiver.
There are no resources available beyond the processes and the link layer as specified above (such as shared storage or a gateway), so the processes entirely rely on themselves and the link layer in order to handle certain tasks.
This implies that there is no way for two disconnected components of processes to connect to each other.

There are two types of \emph{actions} that a protocol can execute.
The first type has the form of a standard procedure $\langle label\rangle (\langle parameters \rangle) \rightarrow \langle commands \rangle$, where $label$ is the unique name of that action, $parameters$ specifies the parameter list of the action, and $commands$ specifies the commands to be executed when calling that action.
Such actions can be called locally (which causes their immediate execution) or remotely.
In fact, we assume that every message must be of the form $\langle label \rangle (\langle parameters \rangle)$, where $label$ specifies the action to be called in the receiving process and $parameters$ contains the parameters to be passed to that action call.
All other messages are ignored by the processes.
The second type has the form $ \langle label\rangle: \langle guard \rangle \rightarrow \langle commands \rangle$, where $label$ and $commands$ are defined as above and $guard$ is a predicate over local variables.
We call an action whose guard is simply \textbf{true} a \emph{timeout} action.

The \emph{system state} is an assignment of values to every variable of each process (including its buffers).
An action in some process $v$ is \emph{enabled} in some system state if its guard evaluates to \textbf{true}, or if there is a message
requesting to call it that was transmitted to the process by the link layer and has not been processed yet.

A \emph{computation} is an infinite fair sequence of system states such that for each state $S_i$, the next state $S_{i+1}$ is obtained by executing an action that is enabled in $S_i$.
This disallows the overlap of action executions, i.e., action executions are \emph{atomic}.
We assume \emph{weakly fair action execution}, meaning that if an action is enabled in all but finitely many states of a computation, then this action is executed infinitely often.
Note that a timeout action of a process is executed infinitely often.
Besides this, we place no bounds on process execution speeds, and no restrictions on the execution order of enabled actions, i.e., we allow fully asynchronous computations and non-FIFO message delivery.

\subsection{Problem statement} \label{sec:FDP}

A protocol is \emph{self-stabilizing} with respect to a set of legitimate states if it satisfies the following two properties: 
\begin{description}
\item[\emph{Convergence:}] starting from an arbitrary system state, the
protocol is guaranteed to eventually arrive at a legitimate state.
\item[\emph{Closure:}] starting from a legitimate state the protocol remains in
legitimate states thereafter.
\end{description}
A self-stabilizing protocol is thus able to recover from transient faults regardless of their nature. Moreover, a self-stabilizing protocol does not have to be initialized as it eventually starts to behave correctly regardless of its initial state.
A formal definition of the FDP can be found in \cite{DBLP:conf/sss/ForebackKNSS14}, which we briefly recap in the following.

\subsubsection{Definition of the Finite Departure Problem}
We assume each process to be either tagged as {\em leaving} or {\em staying}, and this information is available in a read-only Boolean variable called {\em leaving}. A leaving process makes a decision to leave the system by executing the \textbf{stop} command. A process that has executed \textbf{stop} will be {\em inactive} (i.e., none of its actions will be enabled) from that point on, otherwise it is called {\em active}. A process $p$ can \emph{safely} leave a system if the active processes of its weakly connected component are still weakly connected without $p$.

In the \emph{Finite Departure Problem (FDP)} the problem is to eventually reach a state in which
(i) every staying process is active,
(ii) every leaving process is inactive, and (iii) for each
weakly connected component of the initial network, the staying processes
in that component still form a weakly connected component. 
Such a state is called a {\em legitimate} state for the FDP. 
Our goal is to come up with a self-stabilizing local-control protocol for the FDP, i.e., it satisfies the convergence and closure property w.r.t. these legitimate states. 
For the convergence we only consider initial states in which all processes are initially active (as inactive processes are equivalent to non-existing processes).
We also restrict the initial state to contain only a finite number of messages in the buffers.
Without this assumption, there might be an infinite number of corrupted messages, which would make it impossible to guarantee convergence in finite time in general.
Finally, we do not allow the presence of connections to non-existing processes as is normally assumed in self-stabilizing topologies (see, e.g.,~\cite{DBLP:journals/tcs/BernsGP13,DBLP:conf/sss/NorNT13,DBLP:conf/wdag/ScheidelerSS16,DBLP:conf/opodis/ScheidelerSS15,KoutsopoulosSS15}).
Note that this does not make the problem trivial because even under that assumption there is no self-stabilizing local-control protocol for the FDP \cite{DBLP:conf/sss/ForebackKNSS14}.

\subsection{Related work}\label{sec:related_work}
The idea of self-stabilization in distributed computing was introduced in a classical paper by E.W. Dijkstra in 1974~\cite{Dijkstra74}, in which he investigated the problem of self-stabilization in a token ring.
In the past 10 years, several self-stabilizing local-control protocols have been proposed for various overlay networks (e.g., \cite{corona,JRSST09,DolevK08,AspnesW07,JacobRSS2012,DBLP:journals/tcs/BernsGP13}), but none of them has considered the leaving of nodes as an individual problem until the work of Foreback et al. \cite{DBLP:conf/sss/ForebackKNSS14}.
In that work, two problems are considered: the Finite Departure Problem (FDP) and the Finite Sleep Problem (FSP). The authors show that there is no self-stabilizing local-control protocol for the FDP, so oracles are investigated that allow the FDP to be solved. In the FSP, the leaving processes do not have to make an irrevocable decision when to leave the network. They just fall asleep whenever they think it is safe to do so, but they will be woken up again as long as there are still messages in their channel. So the goal of the FSP is just to ensure that eventually a state is reached where all leaving nodes are permanently asleep. Foreback et al. show that for the FSP problem, a self-stabilizing local-control protocol does exist. While their protocol only works together with a self-stabilizing protocol for arranging nodes in a sorted list, more universal protocols for the FSP were presented in \cite{KoutsopoulosSS15}. Another extension of the results in \cite{DBLP:conf/sss/ForebackKNSS14} was presented in \cite{FNT16}. In that paper, the authors study churn in general (including join and leave requests) and consider the case that neither the number of churn requests in total, nor the number of concurrent churn requests can be bounded by a constant.
They prove that a solution to this problem is possible if and only if not every request needs to be satisfied.

Aside from these results, there has been research on self-stabilizing link layers \cite{DBLP:journals/ipl/DolevDPT11,DBLP:conf/sss/DolevHSS12} which even guarantee FIFO-delivery, thus giving stronger guarantees than required for the relay layer.

\subsubsection{Related work concerning relays}\label{sec:related_work_concerning_relays}
To the best of our knowledge, our relay concept has not been proposed before in the distributed systems community. Of course, there is a long history of using relays in communication networks. Relays are commonly used when two devices are too far away from each other to exchange information directly, like in wireless networks, or if two devices cannot interact directly because of firewalls. Relay networks have also been used to improve availability (a prominent example are Resilient Overlay Networks \cite{ABKM01}), to provide anonymity (a prominent example is the TOR network \cite{DMS04}), or to improve performance (a prominent example is AKAMAIs IPA Relay service). In general, most of the peer-to-peer systems and overlay networks proposed so far are using their members as relays for the exchange of requests or information between its members.

Our relay concept also has some interesting connections to the access control domain.
Due to the \textbf{delete} and \textbf{stop} commands, it is possible to grant and revoke access rights with relays.
There is a huge amount of literature on access control in distributed systems.
For surveys on access control approaches in various contexts such as operating systems, file systems, distributed systems, and web-based systems, see e.g. \cite{surveyDelessy07,surveyKirrane08,surveyMiltchev08,surveyHu06,surveyKoshutanski09,surveyLazouski10}.
Important requirements for access control schemes are
\begin{itemize}
\item Integrity: It should not be possible to construct, tamper with or steal an access right.
\item Propagation: There should be mechanisms in place controlling the transfer of access rights.
\item Revocation: It should be possible to revoke an access right.
\end{itemize}
Interestingly, our relay approach can satisfy these requirements if the processes cannot tamper with the relay layer.
 
The simplest ways of controlling access rights are to use passwords or cryptographic keys, but these can easily be delegated from one process to another. Another simple way is to use access control lists (ACLs), which gained prominence in the 1970s with the advent of multiuser systems like UNIX. In distributed systems, the ACL approach usually requires a trusted third party, in order to prevent tampering with the ACLs. Other popular access control models are Role-Based Access Control (RBAC), Attribute-Based Access Control (ABAC), Policy-Based Access Control (PBAC), and Risk-Adaptive Access Control (RAdAC). For all of these models, various variants have been proposed depending on the context in which they are used and the focus on particular properties. In most of the implementations of these models, trusted third parties are used as well since otherwise it is hard to guarantee integrity and prevent uncontrolled propagation of access rights. More decentralized approaches keep track of delegation chains which is somewhat similar to our relay approach: if one of the delegations is revoked, all delegations beyond it will not be accepted any more so that the corresponding processes will lose their access rights to certain objects. However, to the best of our knowledge, this chaining approach has not been used in order to control the interconnection of processes. An example where access rights are provided via explicit communication channels is the Singularity operating system \cite{hunt05,wobber07}. A key aspect of Singularity are Software-Isolated Processes (SIPs), which encapsulate pieces of an application or a system and provide information hiding, failure isolation, and strong interfaces. Communications between SIPs is through bidirectional, strongly typed, higher-order channels. When a channel is created, both of its endpoints are returned to the SIP that created it. These endpoints can be freely delegated along existing channels but not replicated, which provides a more flexible form a access control than our relay approach, but this still opens up the possibility of stealing access rights resp. delegating them by mistake.

\subsection{Our contributions}

We present a novel approach for interconnecting processes which is based on so-called \emph{relays}.
For this we introduce a novel relay layer that acts between the application and the link layer and show that depending on the way an application uses this layer, this relay layer can self-stabilize to a legal state in which a transfer of messages is guaranteed.
After that, we show that our relay approach is universal in a sense that one can get from any weakly connected network to any other weakly connected network while maintaining weak connectivity in the transformation process.
In Section~\ref{sec:adapt_existing_solutions} we show that existing solutions for the FDP can be transformed for our relay layer such that they solve the FDP without the use of an oracle (assuming a reliable link layer instead).
Our relay concept also has some interesting connections to the access control domain as we pointed out in Section~\ref{sec:related_work_concerning_relays}.
\section{The Relay Layer}\label{sec:relay_layer}

We assume that all connections between processes happen through relays which are managed by a so-called {\em relay layer}. Each process $v$ is assumed to interact with its own, separate relay layer $RL(v)$ (so that it is clear which relay is owned by which process), and $RL(v)$ is required to reside at the same machine as $v$ so that interactions between $v$ and $RL(v)$ are local. Whenever a message needs to be sent to $v$, it has to go through $RL(v)$. Each $RL(v)$ has a globally unique address, or short $RID$, that depends on the address of its machine, so that messages can be sent to it from any other relay layer that knows its $RID$.
Furthermore, every relay layer $RL(v)$ has a local buffer $RL(v).Buf$ that is used for the internal communication between relay layers:
Every $RL(v).Buf$ is expected to consist of pairs $(targetRID, message)$ in which $targetRID$ is the $RID$ of the relay layer $message$ is sent to by the link layer.
Here $message$ must be an internal message (any other type of message will be ignored).
The relay layer of the entire system is the set of relay layers over all of its processes. 

\subsection{Relays}

A relay is basically a socket that is non-transferably owned by exactly one process $v$, and that can have both incoming connections from other relays as well as an outgoing connection to some relay. More precisely, $RL(v)$ maintains the following variables for each relay $r$:
\begin{itemize}
\item $r.ID$: globally-unique identifier of relay $r$ (containing the $RID$ of its relay layer so that messages can be sent to $r$ when knowing its ID)
\item $r.state$: is either {\em alive} or {\em dead}
\item $r.out$: stores a $(Key,ID)$ pair where $Key$ is a set keys, and $ID$ is the ID of the target of the outgoing connection (if $ID=\bot$ then $r$ is a {\em sink}, i.e., messages are forwarded to the process owning it)
\item $r.level\in \N_0$: stores the distance of $r$ (in hops) to the sink relay reached via its outgoing connection (there is always a unique such one, see below)
\item $r.sinkRID$: stores the RID of the \emph{sink} of $r$, i.e., the RID of the relay layer of the process that will receive messages sent via $r$
\item $r.In$: set of triples of the form $(key, RID, \bot)$ or $(key,\bot, r')$ for some relay $r'$, where $key$ is a globally unique key (depending on the $RID$ of $r$'s relay layer), $RID$ specifies the address of the relay layer that can send messages to $r$ via $key$, and $r'$ is a relay via which $key$ was supposed to be forwarded; depending on the form, $key$ is a \emph{confirmed} or \emph{unconfirmed} key
\item $r.Buf$: stores all messages that the link layer should send to the relay layer with \RID $r.out.ID$ if $r.out.ID \ne \bot$ or to $v$ if $r.out.ID = \bot$
\end{itemize}
Note that we assume all buffers (i.e., $RL(v).Buf$ and $r.Buf$ for every relay $r$) to be insert-only, i.e., only the link-layer can remove a message from them.
Furthermore, we assume all IDs in the system to be valid, i.e., for every ID in the system the corresponding RID belongs to an existing process (it would be possible to lift this assumption by introducing another oracle or by giving more power to the underlying link layer, but this is beyond the scope of this paper).

The relay connections can be represented by a so-called relay graph.
\begin{definition}
Given any system state $S$, the {\em relay graph} $G=(V,E)$ of $S$ is a directed graph that is defined as follows:
$V=R \cup P$, where $R$ is the set of relays and $P$ is the set of active processes.
$E=E_P \cup E_{Ch}$ where $E_P$ is the set of all {\em explicit} edges and $E_{Ch}$ is the set of {\em implicit} edges. $E_P$ contains an edge $(v,w)$ whenever 
\begin{enumerate}
    \item $v \in P$ and $w \in R$ and $w$ is owned by process $v$, 
    \item $v \in R$ and $w \in R$ and relay $v$ has an outgoing connection to relay $w$ (i.e., $v.out.ID = w.ID$), or
    \item $v \in R$ and $w \in P$, and relay $v$ is a sink relay of process $w$ (i.e., $v.out.ID = \bot$).
\end{enumerate}    
    $E_{Ch}$ contains an edge $(v,w)$ whenever $v \in R$, $w \in R$ and a reference to $w$ is contained in the parameter list of a message in $v.Buf$. Thus, while explicit edges can be used to send messages, implicit edges cannot be used to send messages yet.
\end{definition}
Observe that the third requirement on a legitimate state implies that every relay graph is cycle-free.

\subsection{Relay layer primitives}\label{sec:relay_layer_primitives}

Whenever a process holds a reference to a relay $r$, which we denote by $\myref{r}$, we assume that it is a "dark" reference, i.e., the variables of the relay cannot be accessed by the process. However, the reference can be used by the processes to call a number of primitives offered by the relay layer
(in the following, we assume that all relays mentioned below are {\em owned} by the calling process, i.e., they are or have been created for it by its relay layer --- relays not owned by the calling process will be ignored):
\begin{enumerate}
\item \textbf{new Relay}: returns a reference to a new sink relay $r$ with a globally unique identifier $r.ID$, $r.state=alive$, $r.In=\{ \}$, $r.out=(\{ \},\bot)$, and $r.level=0$.

\item \textbf{delete $\myref{r}$}: prepares the relay referenced by $\myref{r}$ for deletion, in a sense that the relay layer sets $r.In=\{ \}$ and $r.state=dead$. This has the effect that $r$ will not accept any further messages, but $r$ still continues to deliver the messages in $r.Buf$. $r$ is deleted by the relay layer once $r.Buf$ is empty and all relay relay keys sent via $r$ have been confirmed or deleted.

\item \textbf{merge($R$)}: if for all relays $r \in R$, $r.state=alive$, $r.out.ID$ is equal to some common $ID$, $r.level$ is equal to some common $\ell$, $r.sinkRID$ is equal to some common $sinkRID$ and $r.In=\{ \}$, the relay layer creates a new relay $r'$ with new $r'.ID$, $r'.state=alive$, and $r'.out=(Key,ID)$ with $Key=\bigcup_{r \in R} r.out.Key$, $r'.level=\ell$, $r'.sinkRID = sinkRID$, $r'.In=\{ \}$, and $r'.Buf=\bigcup_{r \in R} r.Buf$. Also, all relays in $R$ are deleted. A reference to $r'$ is returned back to the process. (If one of the conditions above is not satisfied, merge does nothing.)

\item \textbf{getRelays}: returns (references to) the current set of all relays owned by $v$ that are still alive.

\item \textbf{incoming($\myref{r}$)}: returns $|r.In|$

\item \textbf{direct($\myref{r}$)}: returns true iff $r.level \le 1$

\item \textbf{is-sink$(\myref{r})$}: returns true iff $r.level=0$

\item \textbf{dead($\myref{r}$)}: returns true iff $r$ does not exist anymore or $r.state=dead$

\item \textbf{same-target($\myref{r_1},\myref{r_2}$)}: returns true iff $r_1.out.ID=r_2.out.ID$

\item \textbf{send($\myref{r},action(parameters))$}: if $r$ is still alive, adds a message of the form $((key,r.ID,r.out.ID),action(parameters'))$ for some arbitrary $key \in r.out.Key$ to $r.Buf$ (where $parameters'$ is an adapted form of $parameters$ explained below), where $(key,r.ID,r.out.ID)$ is called the {\em header} of the message.
\end{enumerate}

\begin{figure}
 \begin{subfigure}[b]{0.45\textwidth}\centering
   \includegraphics{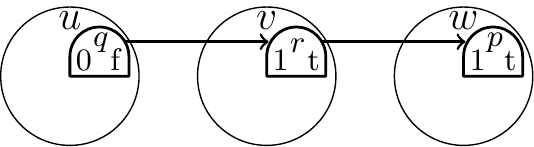}
   \caption{Initial situation. $u$ owns relay $q$, $v$ owns relay $r$ and $w$ owns relay $p$. By definition, $r$ and $p$ are direct relays, whereas $q$ is not.}\label{fig:model:example:a}
 \end{subfigure}
 \hfill
  \begin{subfigure}[b]{0.45\textwidth}\centering
   \includegraphics{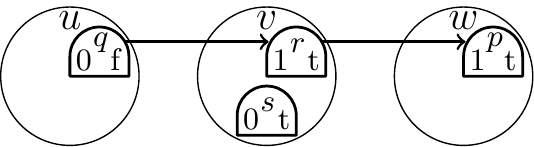}
    \caption{Situation after $v$ has executed \textbf{new Relay}, $v$ has an additional (sink) relay $s$. By definition, $s$ is a direct relay.}\label{fig:model:example:b}
 \end{subfigure}

 \hfill

 \begin{subfigure}[b]{0.45\textwidth}\centering
   \includegraphics{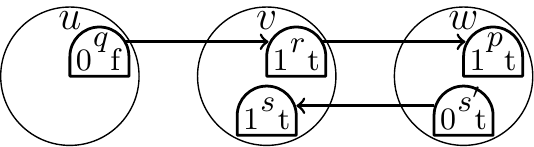}
    \caption{Situation after $v$ has executed $\send{\myref{r}, action(\myref{s})}$ for some action $action$, $RL(w)$ ($w$ is the so-called sink of $r$) has created a new relay $s'$ with an outgoing connection to $s$. 
    }\label{fig:model:example:c}
 \end{subfigure}
 \hfill
 \begin{subfigure}[b]{0.45\textwidth}\centering
   \includegraphics{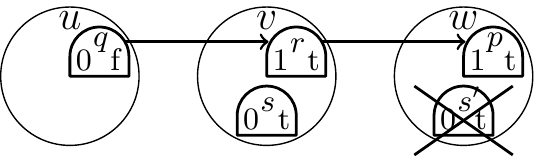}
    \caption{Situation after $w$ has executed \textbf{delete $\myref{s'}$}, $s'$ is marked as dead, $s.In$ has been updated (as $s$ no longer has an incoming connection), and the connection from $s'$ to $s$ has been removed.}
 \end{subfigure}
 \caption{Example with three processes $u$, $v$, and $w$. The characters inside a relay $r$ denote (from left to right), $|r.In|$, the ID of $r$, and whether $r$ is a direct relay. The arrows indicate outgoing connections of relays.}\label{fig:model:example}
\end{figure}
Figure~\ref{fig:model:example} gives examples of the uses of these primitives.

If a process $v$ executes \textbf{stop}, $v$ becomes inactive, and $RL(v)$ immediately deletes all sink relays and from then on periodically deletes all relays $r$ with $r.In = \emptyset$ and $r.Buf = \emptyset$.
$RL(v)$ continues to exist until all relays have been deleted, after which it shuts down.
We hightlight that protocols can prevent relay layers from existing forever by making sure that all indirect relay connections (i.e., relay connections where none of the endpoints is a sink) are closed eventually as we will prove.

Note that the fact that \textbf{merge} can be used to merge relays is the reason for why the variable $r.out.Key$ of a relay $r$ has to be a set instead of a single value only:
The merge could occur in an illegitimate state at which one of the merged relays may store a correct key while another one does not.
At this point it is not clear which one to choose.

For convenience, in the following we will use $RL(r)$ to denote the relay layer that owns a relay $r$, $RID(ID)$ to denote the \RID contained in $ID$, $RID(u)$ to denote the \RID of $RL(u)$ and $RID(r)$ to denote the \RID of $RL(r)$.

\subsection{Message processing and action handling}
All messages that can be sent by a process $v$ are required to be remote method invocations of the form $action(parameters)$ (otherwise, they will be ignored by $RL(v)$). 
More precisely, a process $v$ calls $send(\myref{r}, action(parameters))$ to ask $RL(v)$ to send out a message via $r$.
For simplicity, we assume $parameters$ to consist of a sequence of objects, some of which are relay references, and all other objects do not contain any relay reference at all.
We assume that each action has a fixed number of parameters and specifies which of its parameters are relay references.
When $send(\myref{r}, action(parameters))$ is called for an alive relay $r$, there are two possibilities:
If $r$ is a sink, i.e., $r.out.ID = \bot$, then $action(parameters)$ is put into $r.Buf$ such that the process owning $r$ will receive $action(parameters)$.
Otherwise, $RL(r)$ for every relay reference \myref{s} contained in $parameters$ creates a new globally unique key $key$, inserts $(key, \bot, r)$ into $s.In$ and replaces \myref{s} by the quadruple $(key, s.ID, s.level+1, s.sinkRID)$.
We refer to these quadruples by the term \emph{relay parameter}, the first entry of which is called its \emph{key}, the second is called its \emph{id}, the third is called its \emph{level}, and the fourth its \emph{sinkRID}.
Furthermore we assume that there is a part of each generated key that depends on the generating process and can be used to check whether a key $key$ was generated by a process $u$, in which case we say \emph{$key$ belongs to $u$}. %
Let the list of parameters resulting from the replacements be $parameters'$.
Then, $RL(r)$ puts a \transmit{((key, r.ID, r.out.ID), action(parameters'))} message into $r.Buf$ where $key$ is an arbitrary element from $r.out.Key$.
The pseudocode of the \send{} action can be found in Listing~\ref{lst:rl:send}.

\begin{lstlisting}[basicstyle=\small, mathescape=true,caption=Pseudocode for message processing,label=lst:rl:send, numbers=left, escapechar = |, backgroundcolor = \color{lightgray}]
send($\myref{r}$, action(parameters)) $\rightarrow$
if $r.state = alive$ then
  if $r.out.ID \ne \bot$ then // $r$ is not a sink relay
    let $s_1, \dots. s_k$ denote all parameters of $action$ that are relay 
     references
    for every $i \in \{1, \dots, k\}$ do
      create a new globally unique key $key$
      $s_i.In := s_i.In \cup \{(key, \bot, r)\}$ |\label{line:rl:new_key_added_to_r_in}|
      $s_i' := (key, s_i.ID, s_i.level + 1, s_i.sinkRID)$
      replace $s_i$ in $parameters$ by $s_i'$
    let $key$ be arbitrary such that $key \in r.out.Key$
    $r.Buf := r.Buf \cup \{\transmit{((key, r.ID, r.out.ID), action(parameters'))}\}$
  else
    $r.Buf := r.Buf \cup \{action(parameters)\}$
\end{lstlisting}

We assume that the link layer for every relay $r$ eventually processes every message in $r.Buf$ 
without changing its contents.
The link layer makes sure that every message $m' \in r.Buf$ for a relay $r$ is either processed by the process $v$ owning $r$, in case that $outID = \bot$, or successfully delivered to the process whose relay layer has the RID contained in $r.out.ID$.
After the link layer has processed a message $m'$ in $r.Buf$ for a relay $r$, it removes $m$ from $r.Buf$.

\begin{definition}[Valid message header]
A message $m$ of the form\\ $((key,inID,outID), action(parameters))$ is said to have a {\em valid} header for relay $r$ if $r.ID=outID$, and either $(key,RID, \bot) \in r.In$ with $RID = RID(inID)$, or $(key, \bot, r') \in r.In$
and $r'.sinkRID = RID(inID)$.
\end{definition}

When a message $m = ((key,inID,outID), action(parameters))$ is received by a process $w$, $RL(w)$ acts according to the Pseudocode given in Listing~\ref{lst:rl:ch_added}.
We assume that $\probe{controlKeys, keySequence}$ is a dedicated action type used for the relay layers only, in which $controlKeys$ is a set and $keySequence$ is a sequence of keys.
\begin{lstlisting}[basicstyle=\small,mathescape=true,caption=Pseudocode executed by $RL(w)$ when a message $m$ is received by $w$,label=lst:rl:ch_added, numbers=left, escapechar = |, backgroundcolor = \color{lightgray}]
|\transmit{m = ((key,inID,outID), action(parameters))} $\rightarrow$ |
if there is a relay $r'$ such that and $r'.state = alive$ and $m$ has a valid     
 header for $r'$ then
  if $(key,\bot, r'') \in r'.In$ for some relay $r''$ owned by this process and
   $r''.sinkRID = RID(inID)$ then |\label{line:rl:triple_replacement_begin}|
    // first message received via this connection, activate it          
    $r'.In := r.In \setminus \{(key,\bot, r'')\}$ |\label{line:rl:r_in_remove_message_received}|
    $r'.In := r.In \cup \{(key, RID(inID), \bot)\}$  |\label{line:rl:triple_replacement_end}|
  if $r'.out.ID = \bot$ then // r' is a sink relay
    if $action(parameters) = \probe{controlKeys, keySequence}$ then
      for every $key' \in controlKeys$ do
        if there is no relay $r''$ such that $key' \in r''.out.Key$ then |\label{line:rl:msg_arrived_probefail_check}|          
          let $(key_1, \dots, key_k) = keySequence$
          if there is an $RID$ such that $(key_k, RID, \bot) \in r'.In$ then        
            $RL.Buf := RL.Buf \cup \{ (RID, \probefail{key', (key_1, \dots, key_k)} \}$
    else if all ids of relay parameters of $m$ belong to the 
     same |\RID| $senderRID$ then 
      // (otherwise, the message is obviously corrupted)
      $r'.Buf := r'.Buf \cup \{action(parameters)\}$
      for each relay param $(key', ID', level', sRID')$ in parameters do     
        if there is no relay $r''$ with $key' \in r''.out.Key$ in $RL(w)$ then  |\label{line:rl:receipt_no_double_keys}|
          create a new relay $s$ with:
            $s.ID := newID$, where $newID$ is a new, globally unique 
             ID containing the RID of $RL(w)$
            $s.state := alive$, 
            $s.out := (\{key'\},ID')$, |\label{line:rl:key_added}|
            $s.level := level'$, and
            $s.sinkRID := sRID'$, and
            $s.In := \{ \}$, and         
            $s.Buf := \{ ((key',s.ID,ID'), \probe{\{\}, (key')}) \}$ |\label{line:rl:probe_activation_creation}|
          replace $(key', ID', level', sRID')$ in $parameters$ by $\myref{s}$
        else
          replace $(key', ID', level', sRID')$ in $parameters$ by $\bot$
  else // $m$ needs to be forwarded
    $r'.Buf := r'.Buf \cup \{\transmit{m}\}$
    replace $key$ by an arbitrary $key' \in r'.out.Key$
    replace $inID$ by $r'.ID$
    replace $outID$ by $r'.out.ID$
    if $action(parameters) = \probe{controlKeys, keySequence}$ then |\label{line:rl:probe_forwarding_begin}|
      append $key'$ to $keySequence$
      for every $key'' \in controlKeys$ do
        if there is a message $m' \in r'.Buf$ that contains a relay 
         parameter with key $key''$ then
          remove $key''$ from $controlKeys$ |\label{line:rl:probe_forwarding_end}|                    
else if there is a relay $r'$ s.t. $r'.ID = outID$ and $r'.state = alive$ then 
  // $m$ does not have a valid header for $r'$
  $RL.Buf := RL.Buf \cup \{ (RID(inID), \notauthorized{m}) \}$ |\label{line:rl:messagereceipt_notauthorized}|
else if $outID$ contains the |\RID| of $RL(w)$
  // there is no relay $r'$ s.t. $r'.ID = outID$ and $r'.state = alive$
  $RL.Buf := RL.Buf \cup \{ (RID(inID), \outrelayclosed{outID}) \} $|\label{line:rl:messagereceipt_outrelayclosed}|
\end{lstlisting}
Recall that when a process $v$ calls \send{\myref{r}, m}, and $m$ contains references to relays, $RL(v)$ replaces these references by relay parameters containing the necessary information to establish a connection to these relays.
Additionally $RL(v)$ inserts $(key, \bot, r)$ to $r'.In$ for every relay $r'$ that was contained in this message.
These will be replaced by $(key, RID, \bot)$ after the message has been received by a process. 
To prevent $(key, \bot, r)$ entries in $.In$ sets for which no corresponding messages in the system exist (which would prevent $.In$ from becoming empty after all other relays have been closed), a probing is done via the \probe{} messages to check whether such a message $m$ with a relay parameter with key $key$ exists:
On the path from $r$ to the sink relay, it is checked whether $m$ is contained in the buffer of the next relay on the path.
If this is not the case and the sink does not have a relay with that key, a \probefail{} message will be sent in return to inform $r'$ about this.
Note that the \probefail{} message type contains two parameters: the key that was not found and the sequence of keys that were used to get from the initiator of the \probe{} message to the sink.
The latter is used to find the way back to the initiator via the same path (in reverse order) that the \probe{} message took.
Details of this can be found in Listing~\ref{lst:rl:probefail}.
\begin{lstlisting}[basicstyle=\small,mathescape=true,caption=Pseudocode executed upon \probefail{key, keySequence},label=lst:rl:probefail, numbers=left, escapechar = |, backgroundcolor = \color{lightgray}]
$\probefail{key, keySequence} \rightarrow$
  let $key_1, \dots, key_k = keySequence$ ($k = |keySequence|$)  
  if there is a relay $r$ such that $key_k \in r.out.Key$ then
    if $k > 1$ then
      if there is an $RID$ such that $(key_{k-1}, RID, \bot) \in r.In$ then        
        $RL.Buf := RL.Buf \cup \{ (RID, \probefail{key, (key_1, \dots, key_{k-1})}) \} $
    else
      if there is a relay $r'$ such that $(key, \bot, r) \in r'.In$ then
        $r'.In := r'.In \setminus (key, \bot, r)$ |\label{line:rs:r_in_remove_probefail}|
\end{lstlisting}

When a $\notauthorized{m}$ control message is received and there is a non-sink relay $r$ such that $m$ could have been sent by this, the relay layer removes the key contained in $m$ from $r.out.Key$.
If there is still at least one key left in $r.out.Key$, the message is resent with another key.
Otherwise, all elements $(key, \bot ,r)$ are removed from $r'.In$ for every relay $r'$, and $r$ is deleted.
The pseudocode of this action is given in Listing~\ref{lst:rl:notauthorized}.
\begin{lstlisting}[basicstyle=\small,mathescape=true,caption=Pseudocode executed upon \notauthorized{m},label=lst:rl:notauthorized, numbers=left, escapechar = |, backgroundcolor = \color{lightgray}]
$\notauthorized{m = ((key,inID,outID,level), action(parameters))} \rightarrow$
  if there exists a relay $r$ with $r.ID = inID$, $r.out.ID = outID \ne \bot$, 
   $r.level = level$ and $key \in r.out.Key$ then
    $r.out.Key := r.out.Key \setminus \{key\}$
    if $|r.out.Key| > 0$ then
      replace $key$ in $m$ by an arbitrary $key' \in r.out.Key$
      $r.Buf := r.Buf \cup \{m\}$
    else // outgoing link of $r$ is broken / closed
      // remove all "pending" (unconfirmed) relays sent via r
      for all relays $r'$ do
        for all keys $key$ such that $(key, \bot, r) \in r'.In$ do
          $r'.In := r'.In \setminus \{(key, \bot, r) \}$ |\label{line:rs:r_in_remove_not_authorized}|
      delete $r$ |\label{line:rl:notauthorized_delete}|
\end{lstlisting}

The \timeout action mainly detects and corrects all values that are obviously corrupted and contradict to the definition of a legal state that will be given later.
In addition, for each relay $r$ it serves the following purposes: First, it periodically sends a \ping{r.ID, r.level, r.sinkRID, key} message to every relay layer whose \RID is contained as the second parameter of a triple $(key, RID, \bot)$ in $r.In$.
This is to give connected relays $r'$ with $r'.out.ID = ID$ and $key \in r'.out.Key$ the opportunity to correct their level or sinkRID information and also to determine if there are relays in $r.In$ that do not exist.
Second, it detects and fully removes deleted relays $r$ that do not need to be kept any more (e.g. because all of their messages have been transmitted) and it also shuts down the relay layer if the process is dead and all relays of it have been deleted.
In case $r$ is not a sink, it additionally sends out an \inrelayclosed{r.out.Key, RID(r), r.out.ID} message as to inform the relay layer of the relay with ID $r.out.ID$ that $r$ has been closed.
Third, its sends out the aforementioned \probe{} messages.
The full pseudocode of this action is given in Listing~\ref{lst:rl:timeout}.
\begin{lstlisting}[basicstyle=\small,mathescape=true,caption=Pseudocode of the periodically executed \timeout action,label=lst:rl:timeout, numbers=left, escapechar = |, backgroundcolor = \color{lightgray}]
$true \rightarrow$
  for all $\myref{r} \in getRelays$ do
    if $r.out.ID = \bot$ then 
      $r.level = 0$ |\label{line:rl:timeout_level_set}|
    else 
      if $r.level < 1$ then 
        $r.level := 1$
    if $r.out.ID = \bot$ and $r.out.Key \ne \{ \}$ then 
      $r.out.Key = \{ \}$      |\label{line:rl:timeout_key_empty}|
    if $r.out.ID \ne \bot$ and $r.out.Key = \{ \}$ then 
      delete r |\label{line:rl:timeout_delete_key_empty}|
    for all $(key, X, Y) \in r.In$ do
      if there is a $(key, X', Y') \in r.In$ such that $X' \ne X$ and $Y' \ne Y$ 
       or if there is a $(key, X', Y') \in r'.In$ for some relay $r' \ne r$ or 
       if $key$ does not belong to $RL(r)$ then
        $r.In := r.In \setminus \{ (key, X, Y) \}$ |\label{line:rl:timeout_r_in_remove1}|
    for all $(key, RID, \bot) \in r.In$ do
      $RL.Buf := RL.Buf \cup \{ (RID, \ping{r.ID, r.level, r.sinkRID, key}) \} $ |\label{line:rl:ping_send}|  
    for all $(key, \bot, r') \in r.In$ do
      if $r'$ does not exist then
        $r.In := r.In \setminus (key, \bot, r')$ |\label{todo:this:was:added}|
    for all $x \in r.In$ such that $x \ne (key, RID, \bot)$ and $x \ne (key, \bot, r')$
     for some |\RID| $RID$ and some existing relay $r'$ 
       $r.In := r.In \setminus \{ x \}$ |\label{line:rl:timeout_r_in_remove2}|
    if $r.state = dead$ then
      if there is no relay $r'$ owned by this process such that
       $(key,\bot,r) \in r'.In$ for some key $key$ then
        if $r.out.ID = \bot$ then
          delete $r$ |\label{line:rl:timeout_relay_dead_delete}|
          completely remove $r$
        else if $r.Buf = \emptyset$ then
          let $RID$ be the |\RID| of this relay layer
          let $oID := r.out.ID$
          $RL.Buf := RL.Buf \cup \{ (RID(oID), \inrelayclosed{r.out.Key, RID, oID}) \} $ |\label{line:rl:timeout:inrelayclosed}|
          completely remove $r$
    if the process is dead and $r.In = \{ \}$ and $r.Buf = \{ \}$ and there 
     is no relay $r'$ such that $(key,\bot,r) \in r'.In$ for some key $key$ then
      delete $r$  |\label{line:rl:timeout_process_dead_delete}|  
    if there is a relay $r' \ne r$ such that there is a $key \in r'.out.Key$ 
     such that $key \in r.out.Key$ then
      if $r'.ID > r.ID$ then
        delete r |\label{line:rl:timeout_double_used_key_delete}|
    for every $key \in r.Out.Key$ do
      let $keySequence$ be a sequence consisting of the single 
       element $key$
      let $controlKeys$ be the set of all keys $key'$ s.t. there is a 
       relay $r'$ such that $(key', \bot, r) \in r'.In$ and there is no message 
       in $r.Buf$ containing a relay parameter with key $key'$
      if the corresponding process is alive or $r.In \ne \emptyset$
        $r.Buf := r.Buf \cup \{ ((key, r.ID, r.out.ID), \probe{controlKeys, keySequence}) \}$       
    if the corresponding process is dead and $r.In = \emptyset$ and $r.Buf = \emptyset$
      delete $r$
  if the process is dead and process owns no relay then
    shut down this relay layer completely  
\end{lstlisting}

When a relay layer receives a \ping{ID, level, sinkRID, key} message it checks whether there is a corresponding relay $r$ with $r.out.ID = ID$ and $key \in r.out.Key$. 
If there is no such relay, it responds to the relay layer owning the relay with id $ID$ with an \inrelayclosed{} message indicating that there is no such relay with such a key. 
Otherwise, if $r.level \ge level +1$, it updates $r.level$ to $level$ and $r.sinkRID$ to $sinkRID$.
If $r.level < level+1$, it deletes $r$ (in this case correcting the value would be dangerous as this would allow for cycles in the relay graph).
The pseudocode of the \ping{} message is given in Listing~\ref{lst:rl:ping} and the pseudocode of the $\inrelayclosed{Keys, senderRID,ID}$, which basically removes every entry $(key,RID,\bot)$ from all $.In$ sets such that $key \in Keys$, is given in Listing~\ref{lst:rl:inrelayclosed}.
\begin{lstlisting}[basicstyle=\small,mathescape=true,caption=Pseudocode of the \ping{} action,label=lst:rl:ping, numbers=left, escapechar = |, backgroundcolor = \color{lightgray}]
$\ping{ID, level, sinkRID, key} \rightarrow$
  if $ID \ne \bot$ then
    if there is a relay $r$ with $r.out.ID = ID$ and $key \in r.out.Key$ do
      $r.sinkRID := sinkRID$
      if $r.level > level + 1$ then 
        $r.level := level + 1$
      if $r.level < level + 1$ then 
        delete $r$ |\label{line:rl:ping_delete}|
    else
      let $u$ be the process such that $RID(u) = RID(ID)$
      let $RID$ be the $RID$ of this process
      $RL.Buf := RL.Buf \cup \{ (RID(ID), \inrelayclosed{\{key\}, RID, ID}) \}$ |\label{line:rl:ping:inrelayclosed}|
\end{lstlisting}
\begin{lstlisting}[basicstyle=\small,mathescape=true,caption=Pseudocode of the \inrelayclosed{} action,label=lst:rl:inrelayclosed, numbers=left, escapechar = |, backgroundcolor = \color{lightgray}]
$\inrelayclosed{Keys, senderRID, ID} \rightarrow$
  for every $key \in Key$ do
    if process owns a relay $r$ s.t. $(key, RID, \bot) \in r.In$ 
     for some $RID$ then
      $r.In := r.In \setminus \{(key, RID, \bot)\}$
\end{lstlisting}

When \textbf{delete} \myref{r} is called, $RL(r)$ sets $r.state$ to $dead$ and sends an \outrelayclosed{r.ID} message to every relay layer whose RID is the second parameter of a triple in $r.In$.
Afterwards, it empties $r.In$ so that no message can be received via $r$ from that point in time.
The pseudocode of this is given in Listing~\ref{lst:rl:delete}.
\begin{lstlisting}[basicstyle=\small,mathescape=true,caption=Pseudocode executed upon \textbf{delete} \myref{r},label=lst:rl:delete, numbers=left, escapechar = |, backgroundcolor = \color{lightgray}]
|\textbf{delete} \myref{r}:|
  r.state := dead
  for every $(key, RID, \bot) \in r.In$ do
    let $u$ be the process such that the |\RID| of $RL(u)$ equals $RID$
    $RL.Buf := RL.Buf \cup \{ (RID, \outrelayclosed{r.ID}) \}$ |\label{line:rl:delete_outrelayclosed}|
  $r.In := \{\}$
\end{lstlisting}
Note that a relay $r$ is not closed immediately during the execution of \textbf{delete} \myref{r}.
This is to allow all messages still in $r.Buf$ to be delivered first.
Once this has happened, the relay will be removed completely upon the execution of \timeout.

When a relay layer receives an \outrelayclosed{ID} message and owns a relay $r$ with $r.out.ID = ID$, it removes all triples $(key, \bot, r)$ from $r'.In$ for every relay $r'$ owned by it, empties $r.out.Key$, sets $r.out.ID$ to $\bot$, and calls delete afterwards.
The pseudocode of this action can be found in Listing~\ref{lst:rl:outrelayclosed}.
\begin{lstlisting}[basicstyle=\small,mathescape=true,caption=Pseudocode of the \outrelayclosed{} action,label=lst:rl:outrelayclosed, numbers=left, escapechar = |, backgroundcolor = \color{lightgray}]
$\outrelayclosed{ID} \rightarrow$
  if this process owns a relay $r$ such that $r.out.ID = ID$ then
    // remove all "pending" (unconfirmed) relays sent via r
    for all relays $r'$ do
      for all keys $key$ such that $(key, \bot, r) \in r'.In$ do
        $r'.In := r'.In \setminus \{(key, \bot, r) \}$ |\label{line:rl:r_in_remove_outrelayclosed}|
    $r.out.Key := \{\}$
    $r.out.ID := \bot$
    delete r |\label{line:rl:outrelayclosed_delete}|
\end{lstlisting}

\subsection{Properties of the relay layer}\label{sec:legal_state}
In order to define legal states for the relay layer, we introduce the following notion of a \emph{valid relay}:
\begin{definition}[Valid Relay]\label{def:valid_relay}
A relay $r$ is \emph{valid} iff
\begin{enumerate}
\item $r.state = alive$, and \label{item:vr:state_alive}
\item $r.ID$ is globally unique, and \label{item:vr:rid_unique}
\item $r.out$ stores a pair (Key, ID) such that $Key$ is a set, and \label{item:vr:rout_key}
\item $r.In$ only consists of triples $(key,RID, \bot)$ with $RID \ne \bot$ or $(key, \bot, r'')$ for a valid relay $r''$ owned by $RL(r)$, and \label{item:vr:rin_stores_triples}
\item every key $key$ used as a first parameter of a triple in $r.In$ is locally unique (i.e., it does not appear in any other triple in $r.In$ or $r''.In$ for any relay $r'' \ne r$) and belongs to $RID(r)$, and\label{item:vr:in_key_unique}
\item there is no $\ping{r.ID,level, sinkRID, key}$ message in the system such that $level \ne r.level$ or $sinkRID \ne r.sinkRID$, or $(key, \bot, r'') \in r.In$ for any relay $r''$, and  \label{item:vr:ping}
\item there is no $\outrelayclosed{r.ID}$ message in the system, and \label{item:vr:outrelayclosed}
\item for every $(key, RID, \bot) \in r.In$ there is no\\ \notauthorized{m = ((key,inID,r.ID,level), action(parameters))} message in the system for any $level$ and any $inID$ such that $RID(inID) = RID$, and
	for every $(key, \bot, r'') \in r.In$ there is no\\ \notauthorized{m = ((key,inID,r.ID,level), action(parameters))} message in the system for any $level$ and any $inID$ such that $RID(inID) = r''.sinkRID$, and  \label{item:vr:notauthorized}
\item for every $(key, \bot, r'') \in r.In$,
there is no \probefail{key, (key_1, \dots)} message in the system such that $key_1 \in r''.out.Key$, and,
let $(r_1 = r'', r_2, \dots, r_k)$ be the sequence of relays such that $r_{i+1}.ID = r_i.out.ID$ for all $1 \le i < k$ and $r_k.out.ID = \bot$, then 
either for a relay $r'$ owned by the process with \RID $r''.sinkRID$ such that $r'.out.ID = r$, $key \in r'.out.Key$, $r'.level = r.level+1$, there is a \probe{\{\}, key} message with a valid header in transit to $r$ and there is no \probe{controlKeys, (key_1, \dots)} message such that $key \in controlKeys$ and $key_1 \in r''.out.Key$ in $r'''.Buf$ for any relay $r''' \notin \{r_1, \dots, r_{k-1}\}$, or there is a message $m$ with a valid header for $r_{j+1}$ in $r_j.Buf$ for some $1 \le j < k$ containing a relay parameter with key $key$, and there is no \probe{controlKeys, (key_1, \dots)} message such that $key \in controlKeys$ and $key_1 \in r''.out.Key$ in $r'''.Buf$ for any relay $r''' \notin \{r_1, \dots, r_j\}$. 
In addition to all of these properties \textbf{either} \label{item:vr:newprop}
    \item $r$ is a sink, i.e., $r.out = (\{\},\bot)$, $r.level = 0$, and $r.sinkRID = RID(r)$, \textbf{or} \label{item:vr:sink}
    \item \label{item:vr:non_sink}
    \begin{enumerate}
	\item $r.out.ID \ne \bot$, and \label{item:vr:non_sink:outID}
	\item there is a valid relay $r'$ with $r'.ID = r.out.ID$, and  \label{item:vr:non_sink:nextValid}
	\item $r.level = r'.level + 1$, and $r.sinkRID = r'.sinkRID$, and \label{item:vr:non_sink:level_and_rid_correct} 
	\item there is a $key \in r.out.Key$ such that $(key, RID, \bot) \in r'.In$ and $RID = RID(r)$, or 
	$(key, \bot, r'') \in r'.In$ for a relay $r''$ and $r''.sinkRID = RID(r)$ and $((key,r.ID,r.out.ID), \probe{\{\}, key}) \in r.Buf$, and \label{item:vr:non_sink:one_key_correct}
	\item for every $key \in r.out.Key$, there is no relay $r''' \ne r$ owned by the same process such that $key \in r'''.out.Key$ \label{item:vr:non_sink:out_key_unique}
	\item there is no \inrelayclosed{Keys, RID(r), r.out.ID} message in transit to $RL(r')$ such that $key \in r.out.Key$ for a $key \in Keys$ \label{item:vr:non_sink:inrelayclosed}
\end{enumerate}
\end{enumerate}
\end{definition}

Using this definition, we can define a \emph{valid relay graph} as follows:

\begin{definition}[Valid relay graph]
A \emph{valid relay graph} of a system state $S$ is the subgraph of the relay graph $G=(R \cup P,E_P \cup E_{Ch})$ of S such that every $r \in R$ is valid and every $(v,w) \in E_{Ch}$ is due to a valid relay parameter.
\end{definition}
Note that every valid relay graph is cycle-free due to Property~\ref{item:vr:non_sink:level_and_rid_correct}) of a valid relay.
We say a state $S$ is \emph{legal} if there is no difference between the relay graph of $S$ and its valid relay graph.
Furthermore, we say an application is \emph{deliberate} if it does not delete a relay $r'$ if $r'.In \ne \emptyset$ (note that this includes that it does not call \textbf{stop} as long as there are sink relays with incoming connections).
Given the above definitions, we obtain the following results whose proofs can be found in Appendix~\ref{sec:missing_proofs_legal_state}:
\begin{theorem}\label{thm:rl:messages_sent_via_valid_relay_will_be_received_by_sink}
 If the application is deliberate, every message sent via a valid relay $r$ will be received by the process $u$ with $RID(u) = r.sinkRID$.
\end{theorem}
Thus the process $u$ is also called the \emph{sink process} of $r$.
Observe that in the valid relay graph, every relay $r$ is connected via a directed path to some process $v$, which is the sink process of $r$.

\begin{theorem}\label{thm:rl:valid_relay_initially_will_remain_valid}
 If the application is deliberate, for every computation that starts in a legal state every state is legal.
\end{theorem}

\begin{theorem}\label{thm:rl:fair_delete_no_indirect_forwarding}
 If the application is deliberate, and does not send the reference of an indirect relay (i.e., a relay $r$ such that $direct(r) = false$) and does not send any reference via a relay that is not valid, every computation will reach a legal state.
\end{theorem}

This implies:
\begin{corollary}
 If the application does not issue any commands, starting from any initial state $S$ the system will reach a state $S'$ such $S'$ and every subsequent state are legal.
\end{corollary}
Note that this resembles the classical definition of self-stabilization in which it is assumed that starting from the initial state no change occurs to the system other than by the self-stabilizing protocol.

Since the relay layer of a process that issues \textbf{stop} is not always shutdown immediately, the following is important as well:
\begin{theorem}\label{thm:rl_of_stopped_process_will_eventually_be_destroyed}
 If the application does not keep an indirect relay for an infinite time, all relay layers of inactive processes will eventually be shut down.
\end{theorem}

\section{Universality of the relay approach}\label{sec:universality}
We introduce three rules for the manipulation of edges of a relay graph and show that they are universal, i.e., using them it is possible to get from any arbitrary weakly connected valid relay graph to any other weakly connected valid relay graph involving the same set of processes.
For simplicity, in this section any relay graphs we consider are assumed to be valid relay graphs.
The rules we present are an adaptation of known rules introduced by Koutsopoulos et al.~\cite{KoutsopoulosSS15} (defined below) to our relay model.
In that work, the authors proved these rules to be universal in the common model, which we will rely on in our proofs.
For convenience, in the following, for a relay $r$, we denote the process that stores the sink relay of $r$ as the \emph{sink process of $r$}.
Furthermore, we say a process $u$ has a \emph{relay $r$ to} another process $v$ if $v$ is the sink process of $r$, and $u$ stores $\myref{r}$ in one of its variables or there is a message in transit to $u$ that will cause such a reference to be created upon receipt.
Additionally, a relay $r$ is called a \emph{direct relay} if and only if $\direct(r)$ evaluates to true.
Otherwise, $r$ is called \emph{indirect}.
The set \relayprims of relay rules consists of the following rules:
\begin{description}
\item[Relay Introduction] Assume a process $u$ has a relay $r$ to a process $v$ and another relay $s$ to a process $w$.
  Then $u$ may send $\myref{s}$ to $v$ (via $r$).
\item[Relay Fusion] Assume a process $u$ has two relays $r$ and $r'$ with $\sameTarget(\myref{r},\myref{r}')$. 
    Then $u$ may merge the two relays.
\item[Relay Reversal] Assume a process $u$ has two relays $r$ and $s$ such that $r \ne s$ and $\incoming(r) = 0$. 
  Then $u$ may send $\myref{s}$ via $r$ and subsequently delete $r$.
\end{description}
Examples of these rules are presented in Figure~\ref{fig:relays:prims}.
\begin{figure}
 \begin{subfigure}[b]{\textwidth}\centering
   \begin{minipage}{0.48\textwidth}\centering
   \includegraphics{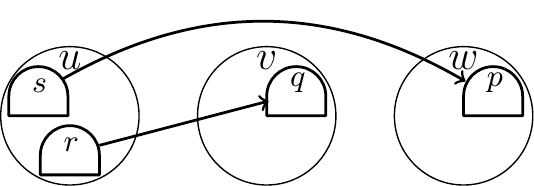}
   \end{minipage}
   \hfill
   \begin{minipage}{0.40\textwidth}\centering
   \includegraphics{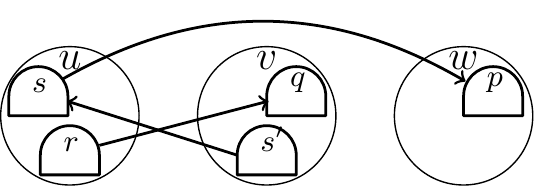}
   \end{minipage}
   \centering
   \caption{Left: Initial situation for \textbf{Relay Introduction}. Right: After $u$ has sent $\myref{s}$ to $v$ (via $r$). }
 \end{subfigure}
 
 \hfill
 
 \begin{subfigure}[b]{\textwidth}\centering
   \begin{minipage}{0.48\textwidth}\centering
   \includegraphics{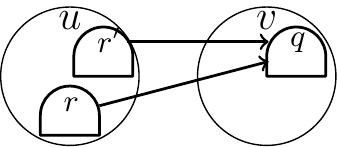}
   \end{minipage}
   \hfill
   \begin{minipage}{0.40\textwidth}\centering
   \includegraphics{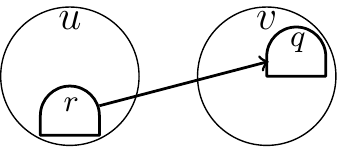}
   \end{minipage}
   \centering
   \caption{Left: Initial situation for \textbf{Relay Fusion}. %
   Right: After $u$ has deleted $r'$. }
 \end{subfigure}
 
  \hfill

   \begin{subfigure}[b]{\textwidth}\centering
   \begin{minipage}{0.48\textwidth}\centering
   \includegraphics{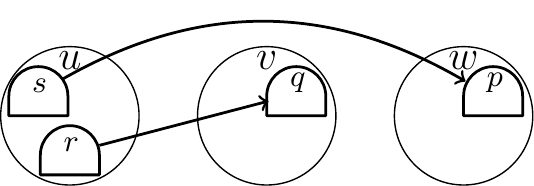}
   \end{minipage}
   \hfill
   \begin{minipage}{0.40\textwidth}\centering
   \includegraphics{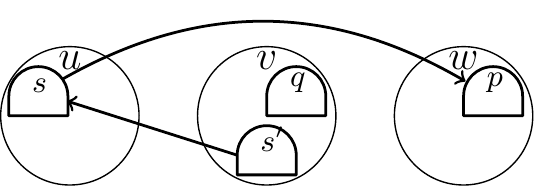}
   \end{minipage}   
   \centering
   \caption{Left: Initial situation for \textbf{Relay Reversal}. Right: After $u$ has sent $\myref{s}$ to $v$ (via $r$) and deleted $r$.}\label{fig:relays:primsC}
 \end{subfigure}

 \caption{Visualization of the rules in \relayprims.}\label{fig:relays:prims}
\end{figure}
The following is easy to show:
\begin{theorem}\label{thm:connectivity}
 \relayprims preserves weak connectivity, i.e., if any of the rules is applied to a weakly connected relay graph $G$, then the resulting graph $G'$ is also weakly connected.
\end{theorem}
The idea of the proof is as follows: 
Relay Introduction does not delete any relay, thus its application cannot harm the connectivity of the relay graph.
Relay Fusion only merges redundant relays.
Last, Relay Reversal preserves weak connectivity because although $u$ deletes
a connection to the sink process of $r$, the message sent causes an edge from $r$ to $s$ (and thus there is an undirected path from $u$ to the sink process of $r$), see Figure~\ref{fig:relays:primsC}.

The universality of the three relay rules is given by the following theorem.
\begin{theorem}\label{thm:universality}
 The rules in \relayprims are universal in a sense that one can get from any weakly connected relay graph $G=(V, E)$ to any other weakly connected relay graph $G'=(V, E')$, where w.l.o.g. $E$ and $E'$ consist solely of explicit edges.
\end{theorem}
The proof of Theorem~\ref{thm:universality} will make of use of the universality of the (process) rules introduced by Koutsopoulos et al.~\cite{KoutsopoulosSS15}, which are restated in the following:
\begin{description}
\item[Introduction]  If a process $u$ has a reference of two processes $v$ and $w$ with $v \neq w$, $u$ \emph{introduces} $w$ to $v$ if $u$ sends a message to $v$ containing a reference of $w$ while keeping the reference.
\item[Delegation] If a process $u$ has a reference of two processes $v$ and $w$ s.t. $u,v,w$ are all different, then $u$ \emph{delegates} $w$'s reference of $v$ if $u$ sends a message to $v$ containing a reference of $w$ and deletes the reference of $w$.
\item[Fusion] If a process $u$ has two references $v$ and $w$ with $v=w$, then $u$ \emph{fuses} the two references if it only keeps one of these references.
\item[Reversal] If a process $u$ has a reference of some other process $v$, then $u$ \emph{reverses} the connection if it sends a reference of itself to $v$ and deletes its reference of $v$.
\end{description}

In the following, we say a \emph{simple relay graph} is a relay graph $G = (P \cup R, E)$ such that all edges in $E$ are explicit and all relays in $R$ are direct and such that every sink relay in $R$ has exactly one incoming connection.
For such a graph, we define the \emph{corresponding process graph} as the multigraph $CPG(G) = (P, E')$ whose vertices are the processes only and whose edge set contains an edge $(u,v)$ with $u,v \in P$ for every edge $(u', v')$ in $G$ such that $u', v' \in R$ and $u'$ is stored in a variable of $u$ and $v'$ is stored in a variable of $v$.
Note that there is a one-to-one relationship between a simple relay graph and its corresponding process graph, i.e., given a process graph $G_P$, there is a (except for isomorphism) unique relay graph $G_R$ with $CPG(G_R)=G_P$.
We say a set of relay rules $RP$ emulates a process rule $p$ if for every simple relay graph $G_R$ and every possible application of $p$ to $CPG(G_R)$, for each resulting process graph $G_P'$ there is a simple relay graph $G_R'$ with $CPG(G_R')=G_P'$ that can be obtained by applying rules from $RP$ only.
This definition enables us to state the following lemma:
\begin{lemma}\label{lem:emulation}
 The relay rules in \relayprims emulate each of the process rules Introduction, Delegation, Fusion, and Reversal.
\end{lemma}
\begin{proof}
The proof strategy is the same for every of the four rules:
Let $G_R$ be an arbitrary simple relay graph.
Further, let $G_P = CPG(G_R)$.
We will then consider an arbitrary application of the particular rule to $G_P$ and denote the resulting graph by $G_P'$.
After that, we show that by applying rules from \relayprims to $G_R$, it is possible to obtain a graph $G_R'$ with $CPG(G_R') = G_P'$.
Thus, in the following we will use these variable names.

We start with the Introduction rule.
Applying the Introduction rule means that for some process $u$ with references of two other processes $v$ and $w$ in $G_P$, $u$ sends a message to $v$ containing a reference of $w$ and keeps the reference.
Thus, in the resulting graph $G_P'$, there is an additional edge $(v,w)$.
In $G_R$, let $u$ send its relay to $v$ to $w$, which resembles a Relay Introduction.
Subsequently, let $w$ create a new relay, send this via the received relay, and close the received relay.
Then this resembles a Relay Reversal.
In the resulting relay graph $G_R'$, all that has changed compared to $G_R$ is that $v$ now has an additional direct relay to $w$.
Thus, in $CPG(G_R')$ all that has changed compared to $CPG(G_R)$ is that there is an additional edge $(v,w)$, thus this graph is isomorphic to $G_P'$.

Next, we deal with the Delegation rule.
Applying the Delegation rule means that for some process $u$ with references to two processes $v$ and $w$ in $G_P$ s.t. $u,v,w$ are all different, $u$ sends a message to $v$ containing a reference a reference of $w$ and deletes the reference of $w$.
Thus, the resulting graph $G_P'$ differs from $G_P$ in that there is an additional edge $(v,w)$ and the edge $(u,w)$ is removed.
In $G_R$, let $u$ send its relay to $v$ to $w$ and delete the relay to $w$.
This resembles the Relay Reversal rule.
After that, let $w$ create a new relay, send this via the received relay, and close the received relay.
Then this resembels the Relay Reversal rule, again.
In the resulting relay graph $G_R'$, all that has changed compared to $G_R$ is that $v$ now has an additional direct relay to $w$ and $u$ no longer has its relay to $w$.
Thus $CPG(G_R') = G_P'$, again.

For the Fusion rule, it is obvious that Safe Fusion emulates the process rule Fusion.

Last, applying the Reversal rule means that some process $u$ that has a reference of some other process $v$, sends a reference of itself to $v$ and deletes its reference of $v$.
In the relay graph, $u$ would create a new relay, send it via the relay to $v$ and subsequently delete its relay to $v$, which resembles a Safe Reversal.
This finishes the proof. 
\end{proof}

Applying this lemma we can finally prove Theorem~\ref{thm:universality}:
\begin{proofof}{Theorem~\ref{thm:universality}}
The idea of the proof is the following:
Assume graph $G = (R \cup P, E)$ is weakly connected. 
First, we show how to transform $G$ into a simple relay graph $G_1$ over the same processes that is weakly connected as well. 
The universality of the process rules from~\cite{KoutsopoulosSS15} and Lemma~\ref{lem:emulation} imply that it is possible to transform this graph into another simple relay graph $G_2$ with $CPG(G_2)=(P,E_2)$ by using the rules in \relayprims, which is defined such that $(w,v) \in E_2$ if and only if in $G'$ there is an edge $(r,s)$ such that $r$ is stored by $v$ and $s$ is stored by $w$.
Last, we show how to transform $G_2$ into a graph isomorphic to $G'$, which finishes the proof.

To transform $G$ into $G_1$, we proceed as follows:
As long as there is still an indirect relay $r$ stored by any process $v$ with $\incoming(r) = 0$, $v$ applies Relay Reversal as follows: $v$ creates a new relay $r'$, sends this relay via $r$ and subsequently closes $r$.
This strictly decreases the number of indirect relays in every iteration.
Note that as soon as there is no indirect relay $r$ with $\incoming(r) = 0$ any more, there cannot be any indirect relay at all (recall that there cannot be any cycles in the sequences of relay connections).
Note that by Theorem~\ref{thm:connectivity}, since we applied Relay Reversal only, the resulting graph $G_1$ is still weakly connected.
Furthermore, $G_1$ is a simple relay graph. 
Thus, as described above, it is possible to transform this graph into a graph $G_2$ as described above.

To transform $G_2$ into $G'$, consider an arbitrary sink relay $s$ in $G'$ and let $T$ be the subgraph of $G'$ that contains all relays $r$ with sink relay $s$.
Note that $T$ is a tree because we assume the relay layer to be in a legitimate state.
Thus, for a relay $r$ in $T$, define $children_T(r)$ as the set of relays $s$ with an edge $(s,r)$ in $T$.
Similarly, for a relay $r \ne s$ in $T$, define $parent_T(r)$ as the relay $s$ for which there is an edge $(r,s)$ in $T$.
By the definition of $G_2$, every process storing a relay $r$ of $T$ (in $G'$) stores (in $G_2$) a direct relay to each process storing a relay $r' \in children_T(r)$ (in $G'$).
Denote by $L_T(i)$ the set of relays at level $i$ of $T$.
First of all, the process storing $s$ (the root of $T$) in $G'$ creates a new relay $r$ (which in the end will be the equivalent of $s$).
Then, it sends $r$ to each process storing a relay $r' \in children_T(r)$ (in $G'$) and closes each of the relays to a process storing a relay $r' \in children_T(r)$ (in $G'$), i.e., the relays via which the relay was sent, thus performing a Relay Reversal.
This way, every process storing a relay $r'$ in $L_T(1)$ (in $G'$) receives a relay $r''$ whose endpoint is equivalent to $parent_T(r'')$.
Then, for every level $i \geq 1$ ascending, every process storing a relay $r'$ in $L_T(i)$ sends the relay it received to each of the processes storing relays in $children_T(r')$ closes the relays via which it sent this relay, thus performing a Relay Reversal, too.
Similarly, every process storing a relay in $L_T(i+1)$ receives a relay $r'$ whose endpoint is equivalent to $parent_T(r'')$.
In the end, we obtain the desired tree $T$.
Since $s$ and thus also $T$ was chosen arbitrarily and since there is no edge in $G_2$ that is not removed in the transformation, this finishes the proof.
\end{proofof}

Recall that we dealt with valid relay graphs in this section.
Luckily, by Theorem~\ref{thm:rl:fair_delete_no_indirect_forwarding} one can show:
 For every protocol that uses only the primitives in \relayprims for the manipulation of edges in the relay graph and only uses references of direct relays in introductions, the underlying relay layer will self-stabilize, i.e., it will reach a state $S$ such the relay graph of $S$ is equal to the valid relay graph of $S$ and starting from any such state for every subsequent state $S'$ the relay graph of $S'$ will be equal to the valid relay graph of $S'$.

\subsection{How to adapt classical protocols to the relay model}\label{sec:adapt_existing_solutions}
One of the benefits of the relay model is that a wide range of protocols designed for the standard interconnection model (e.g., \cite{self-stabilizing-list2,JRSST09,JacobRSS2012,self-stabilizing-list,ShakerR05}) can be adapted to the relay model.
In \cite{DBLP:conf/wdag/ScheidelerSS16} it was shown that a wide range of protocols for static strongly-connected topologies that preserve weak-connectivity can be transformed such that the interaction between nodes can be decomposed into the rules Introduction, Delegation, and Fusion (Theorem~1 of that work).
Lemma~\ref{lem:emulation} of this work yields that these rules can be emulated by the relay rules in \relayprims such that the resulting graph consists of direct relays only.
Putting this together, the aforementioned class of protocols can be adapted to the relay model.
\section{Conclusion and Outlook}
We introduced Relays, a new model for the interconnection of processes in a network.
While we motivated the introduction of this model by that the FDP can be solved in this model, our model has numerous additional advantages, which is why we think it has great potential for future research.

For instance, observe that our model offers facilities for admission control that the common interconnection model does not offer: 
In the standard model, the possession of a reference to another process $u$ admits the sending of a message to $u$ and $u$ is unable to revoke this right.
In the relay model, in contrast, each process is able to delete a relay and thus revoke the right to send a message to this relay.
Even worse in the standard model, a reference can be copied and introduced to other processes without permission of $u$.
Although it is possible to forward a relay reference also in the relay model, such an action only establishes an indirect connection.
To create a direct relay connection, permission of $u$ would still be required.
This has a huge advantage in the scenario of Distributed Denial-of-Service attacks if we assume that the attacker has no access to the relay layer (which may be reasonable if the relay layer is implemented using secure hardware):
If an attacker $v$ forwards a relay reference $r$ to other processes in order to attack the sink process of $r$, the bandwidth of the attack is not the sum of the individual bandwidths of all participating processes, but limited by the bandwidth of $v$.
Thus, forwarding the relays does not yield any advantage for the attack.

Another advantageous property of the relay model whose power is yet to be determined is the fact that processes can create multiple relays as pseudonyms.
Whereas in the original model, each process was uniquely defined by its reference that was even propagated to the application layer, in the relay model applications only know locally valid references.
Although applications can check whether two relays have the same next target, it is not possible for them to determine whether the sink process of two relays with different next targets is equal.
This way, a process could use different sink relays for different purposes in the network and no other processes would be able to link this process's activities, thus achieving anonymity.

\appendix

\section{Proofs of the theorems from Section~\ref{sec:legal_state}}\label{sec:missing_proofs_legal_state}
In this section we present some lemmata and proofs not contained in the main part due to space contraints.

\begin{lemma}\label{lem:delete_never_called}
 If the application does not delete a relay $r'$ if $r'.In \ne \emptyset$, for every alive relay $r$ such that (a) $r.out.ID = \bot$, or (b) Property~\ref{item:vr:non_sink} holds, \textbf{delete} $r$ is only called when $r$ is deleted by the application.
\end{lemma}
This lemma implies that the interal calls of \textbf{delete} $r$ according to the pseudocode do not occur on a valid relay as long as the above prerequisites are fulfilled.
\begin{proof}
For the proof, we check all lines that contain a call of \textbf{delete}:
\begin{enumerate}
 \item Line~\ref{line:rl:notauthorized_delete}: If this line is executed, this means that the relay layer owning $r$ has received a $\notauthorized{m = ((key,inID,outID,level), action(parameters))}$ message with $r.ID = inID$, $r.out.ID = outID \ne \bot$, $r.level = level$ and $key \in r.out.Key$. In this case, we are in Case~(b) (since $r.out.ID = outID \ne \bot$). Due to Property~\ref{item:vr:non_sink:one_key_correct}) and the fact that Property~\ref{item:vr:notauthorized} holds for $r'$ (since $r'$ is valid), the receipt of this message cannot have caused $|r.out.Key|$ to become zero after the removal of $key$ and so this call of \textbf{delete} cannot occur.
 \item Line~\ref{line:rl:timeout_delete_key_empty}: If this line is executed, Property~\ref{item:vr:non_sink:one_key_correct}) must have been violated, yielding a contradiction.
 \item Line~\ref{line:rl:timeout_relay_dead_delete}: If this line is executed, $r.state = dead$ before, yielding a contradiction.
 \item Line~\ref{line:rl:timeout_process_dead_delete}: If this line is executed, $u$ is dead, yielding a contradiction.
 \item Line~\ref{line:rl:timeout_double_used_key_delete}: If this line is executed, $r.out.ID \ne \bot$ and Property~\ref{item:vr:non_sink:out_key_unique}) must have been violated before, yielding a contradiction.
 \item Line~\ref{line:rl:ping_delete}: If this line is executed, $r.out.ID \ne \bot$ and for $r$ Property~\ref{item:vr:non_sink:level_and_rid_correct}) or for the relay $r'$ with $r'.ID = r.out.ID$ Property~\ref{item:vr:sink} and Property~\ref{item:vr:ping} must have been violated, yielding a contradiction.
 \item Line~\ref{line:rl:outrelayclosed_delete}: If this line is executed, $r.out.ID \ne \bot$ and for the relay $r'$ with $r'.ID = r.out.ID$ Property~\ref{item:vr:sink} and Property~\ref{item:vr:outrelayclosed} must have been violated, yielding a contradiction. 
\end{enumerate} 
\end{proof}

\begin{lemma}\label{lem:valid_relay_remains_valid}
	If the application does not delete a relay $r'$ if $r'.In \ne \emptyset$ and does not send any reference via a relay that is not valid, every valid relay $r$ remains valid as long as $r$ is not deleted by the application (including being merged).
\end{lemma}
\begin{proof}
Assume that the application does not delete a relay $r'$ if $r'.In \ne \emptyset$, and let $r$ be the valid relay that is the first to become invalid without being deleted itself.
Denote by $u$ the process owning $r$. 
We handle all different cases that may cause $r$ to become invalid.

Note that $r.state$ is only changed when \textbf{delete} \myref{r} is called.
If this is done by the application, there is nothing to be proven.
According to Lemma~\ref{lem:delete_never_called} and the fact that $r$ is valid, there is no other occasion at which this happens.
Thus Property~\ref{item:vr:state_alive} holds.

Next Check that $r.ID$ and the form of $r.out$ are never changed for any existing relay,
thus Property~\ref{item:vr:rid_unique} and Property~\ref{item:vr:rout_key} hold.

Regarding Property~\ref{item:vr:rin_stores_triples}, note that if $(key,RID,\bot)$ is added to $r.In$ then $RID \ne \bot$.
Furthermore, note that when a new $(key, \bot, r'')$ is added to $r.In$, then $r''$ is owned by $RL(r)$ and when this happen, a relay reference is sent via $r''$.
By the assumption $r''$ is thus a valid relay and the property holds.
There are no other types of elements added to $r.In$.

For Property~\ref{item:vr:in_key_unique} note that the only occasion at which a key becomes the first parameter of an element in $r.In$ that has not been a first parameter of an element in $r.In$ before, is in Line~\ref{line:rl:new_key_added_to_r_in} in which $key$ has just been uniquely generated.

For Property~\ref{item:vr:ping} note that a \ping{r.ID, level, sinkRID, key} message is sent only during \timeout and only with $level = r.level$ and $sinkRID = r.sinkRID$ and $key$ such that $(key, RID, \bot) \in r.In$ which by Property~\ref{item:vr:in_key_unique} prevents $(key, \bot, r'') \in r.In$ to hold.

For Property~\ref{item:vr:outrelayclosed}, note there are two occasions at which an \outrelayclosed{ID} message is sent: In Line~\ref{line:rl:messagereceipt_outrelayclosed} in Listing~\ref{lst:rl:ch_added} and in Line~\ref{line:rl:delete_outrelayclosed} in Listing~\ref{lst:rl:delete}.
In both cases, the \RID contained in $ID$ is the \RID of the sending process and the relay with id $ID$ either does not exist or is dead.
So, if such a message is sent, $r$ cannot have been valid before.

For Property~\ref{item:vr:notauthorized}, check in the pseudocode that the only occasion at which a \notauthorized{m = ((key,inID,r.ID,level), action(parameters))}  message is only sent in Line~\ref{line:rl:messagereceipt_notauthorized}, in which case $m$ did not have a valid header for $r$, i.e., $key$ was not the first parameter of any triple in $r.In$.
Furthermore, note that whenever a new key $key$ is added as the first parameter of a triple to $r.In$, $key$ has just been created as a globally unique key.
Thus, adding an element to $r.In$ also does not violate Property~\ref{item:vr:notauthorized}.

For Property~\ref{item:vr:newprop}, first note that according to the pseudocode for a \probefail{key, (key_1, \dots)} message to be created, a \probe{controlKeys, (key_1, \dots)} message such that $key \in controlKeys$ must have been received by a sink relay $r'''$.
By Property~\ref{item:vr:newprop} this must have been a relay owned by the process with \RID $r''.sinkRID$ and this process must have a relay $r'$ such that $key \in r'.out.Key$, in which case the \probefail{} message is not sent (see Line~\ref{line:rl:msg_arrived_probefail_check} and thereafter).
Next consider an arbitrary triple $(key, \bot, r'') \in r.In$ for a relay $r'$ owned by the process with \RID $r''.sinkRID$ such that $r'.out.ID = r$, $key \in r'.out.Key$, $r'.level = r.level+1$, and there is a \probe{\{\}, key} message with a valid header in transit to $r$.
Note that either this message remains in transit or is received by $r$, in which case the tuple would be removed from $r.In$ (in which case we are finished).
Furthermore, note that no \probe{controlKeys, keySequence} message such that $key \in controlKeys$ is ever put into $r'''.Buf$ for some relay $r''' \notin \{r_1,\dots,r_{k-1}\}$ because of Lines~\ref{line:rl:probe_forwarding_begin}-\ref{line:rl:probe_forwarding_end} in Listing~\ref{lst:rl:ch_added} and the fact that Property~\ref{item:vr:newprop} held before.
Thus, now consider an arbitrary triple $(key, \bot, r'') \in r.In$ such that for the sequence of relays $(r_1 = r'', r_2, \dots, r_k)$, such that $r_{i+1}.ID = r_i.out.ID$ for all $1 \le i < k$ and $r_k.out.ID = \bot$, there is a message $m$ with a valid header for $r_{j+1}$ in $r_j.Buf$ for some $1 \le j < k$ containing a relay parameter with key $key$, and there is no \probe{controlKeys, (key_1, \dots)} message in $r_l.Buf$ for any $j < l < k$ such that $key \in controlKeys$ and $key_1 \in r''.out.Key$.
Note that Property~\ref{item:vr:non_sink:nextValid} implies that all $r_i$ are valid for $1 \le i \le k$.
If the message $m$ is received by $r_{j+1}$ then either $j+1 < k$, in which case the message will be put into $r_{j+1}.Buf$ and the property still holds, or $j+1 = k$, in which case $r_k$ is a sink owned by the process with \RID $r''.sinkRID$ and upon receipt of $m$, that relay layer will generate a \probe{\{\}, key} message and send it to $r$ (see Line~\ref{line:rl:probe_activation_creation} in Listing~\ref{lst:rl:ch_added}).
The only missing part for Property~\ref{item:vr:newprop} is that no \probe{controlKeys, (key_1, \dots)} message such that $key \in controlKeys$ and $key_1 \in r''.out.Key$ is ever put into $r'''.Buf$ for some relay $r''' \notin \{r_1,\dots,r_j\}$. 
This, again, is ensured by Lines~\ref{line:rl:probe_forwarding_begin}-\ref{line:rl:probe_forwarding_end} in Listing~\ref{lst:rl:ch_added} and the fact that Property~\ref{item:vr:newprop} held before.

Assume $r$ is a sink and Property~\ref{item:vr:sink} holds true.
Note that $r.out$ is never changed for a sink relay with $r.out.Key = \{\}$.
Additionally, check that $r.level$ is only changed if $r.out.ID = \bot$ and $r.level \ne 0$, or $r.out.ID \ne \bot$.
Furthermore, $r.sinkRID$ is only changed if $r.out.ID = \bot$ and $r.sinkRID \ne sinkRID(r)$ or $r.out.ID \ne \bot$.
Thus, in this case, the lemma follows.

Now assume $r$ is valid and $r.out.ID \ne \bot$.
Recall that since $r.out.ID$ is never changed, Property~\ref{item:vr:non_sink:outID}) remains true.

We now consider Property~\ref{item:vr:non_sink:nextValid}):
Assume the relay $r'$ with $r'.ID = r.out.ID$ becomes invalid.
Since $r$ is the first relay to become invalid without being deleted by the application, $r'$ must have been deleted by the application for this to occur.
By the restriction on deletions, however, $r'.In = \emptyset$ must have held.
This, however, contradicts Property~\ref{item:vr:non_sink:one_key_correct}) for $r'$.

Next assume that a violation of Property~\ref{item:vr:non_sink:level_and_rid_correct}) causes $r$ to become invalid.
In the case we consider, this only happens if $RL(u)$ receives a $\ping{ID, level, sinkRID, key}$ message with $ID = r.out.ID$ and $r.level \ne level +1$ or $r.sinkRID \ne sinkRID$, which would contradict either Property~\ref{item:vr:ping} for $r'$ according to Property~\ref{item:vr:non_sink:nextValid}) or Property~\ref{item:vr:non_sink:level_and_rid_correct}.

For Property~\ref{item:vr:non_sink:one_key_correct}) check that a $key$ is only removed from $r.out.Key$ if $RL(u)$ receives a $\notauthorized{m = ((key,inID,outID,level), action(parameters))}$ message with $r.ID = inID$, $r.out.ID = outID \ne \bot$, $r.level = level$ and $key \in r.out.Key$.
For those keys $key \in r.out.Key$ such that $(key,RID(r), \bot) \in r'.In$ or $(key, \bot, r'') \in r'.In$ and $r''.sinkRID = RID(r)$, according to Property~\ref{item:vr:non_sink:nextValid}), this cannot happen due to Property~\ref{item:vr:notauthorized} for $r'$.
Now consider the case that the last triple $(key, \dots, \dots) \in r'.In$ such that $key \in r.out.Key$ that fulfills Property~\ref{item:vr:non_sink:one_key_correct}) is removed from $r'.In$.
According to the pseudocode and the fact that $r'$ is valid,
if the triple is of the form $(key, RID(r), \bot)$
this only happens when 
$RL(r')$ receives an $\inrelayclosed{Keys, RID(r)}$ message with $key \in Keys$ (and by Property~\ref{item:vr:in_key_unique}, $key$ belongs to $RL(r')$ because it was stored in $r'.In$).
According to Property~\ref{item:vr:non_sink:inrelayclosed}) such a message cannot have been in transit to $RL(r')$ before.
If, however, the triple is of the form $(key,\bot,r'') \in r'.In$, then according to the pseudocode either the triple is replaced by a triple $(key, RID, \bot)$ such that $RID = r''.sinkRID = RID(r)$ (see Lines \ref{line:rl:triple_replacement_begin}-\ref{line:rl:triple_replacement_end}), in which case the property still holds, or $RL(r')$ receives a \probefail{key, (key_1, \dots)} message such that $key_1 \in r''.out.Id$, which cannot be due to Property~\ref{item:vr:newprop}.
Thus, the property still holds.

For Property~\ref{item:vr:non_sink:out_key_unique}), note the check in Line~\ref{line:rl:receipt_no_double_keys} of Listing~\ref{lst:rl:ch_added} and the fact that only in the subsequent line, a new key is added to $r'out.Key$ for any relay $r'$.

For Property~\ref{item:vr:non_sink:inrelayclosed}) note that an \inrelayclosed{key,RID(r), ID} message with $key \in r.out.Key$ is only sent by $RL(r)$ and only if either $r$ is dead (according to Property~\ref{item:vr:non_sink:out_key_unique}), which is a contradiction, or $RL(r)$ received a \ping{ID, level, sinkRID, key} message such that $ID \ne r.out.ID$ or $key \notin r.out.Key$, but then the message would not violate Property~\ref{item:vr:non_sink:inrelayclosed}.

This concludes the proof of the lemma.
\end{proof}

\begin{definition}[Valid relay parameter]
A relay parameter $(key, ID,level,sinkRID)$ contained in a message $m$ in a buffer $r.Buf$ is valid iff
\begin{enumerate}
  \item $r$ is valid, and
  \item there is no other relay parameter with key $key$, and
  \item $ID \ne \bot$ and the relay $r'$ with $r'.ID = ID$ is valid, and
  \item $level = r'.level + 1$, and $sinkRID = r'.sinkRID$, and
  \item $(key, \bot, r'') \in r'.In$ for some valid relay $r''$ owned by the same process as $r'$ and $r''.sinkRID = r.sinkRID$, and $key$ belongs to $RL(r')$ and
  \item all $ID$s contained in relay references in $m$ belong to the same \RID, and
  \item there is no relay $r'''$ in the system such that $key \in r'''.out.Key$, and
  \item there is no message $m$ in the system using key $key$, and
  \item there is no \probefail{key, (key_1, \dots, key_k)} message in the system with $key_1 \in r''.out.Key$
  \item for every \probe{controlKeys, (key_1, \dots, key_k)} message $m'$ in the system with $key \in controlKeys$ such that for the first element $key_1$ in $keySequence$, $key_1 \in r''.out.Key$, there is a sequence of relays $(r_1 = r'', r_2, \dots, r_k, \dots, r_s)$, $s \ge k+2$ such that $r_{i+1}.ID = r_i.out.ID$ for all $1 \le i < s$ and $key_j \in r_j.out.Key$ for all $1 \le j \le k$ and $m'$ is stored in $r_k.Buf$, and $m \notin r_l.Buf$ for all $1 \le l \le k$, and $m \in r_{i'}.Buf$ for some $k < i' < s$, and
  \item there is no \inrelayclosed{Keys, RID, ID} message in the system with $key \in Keys$ and $ID = r'.ID$ 

\end{enumerate}
 \end{definition}

\begin{lemma}\label{lem:valid_relay_parameter_created_is_valid}
 If the application does not delete a relay $r'$ if $r'.In \ne \emptyset$, each relay parameter created from a valid relay by sending a message via a valid relay is a valid relay parameter.
 Furthermore, every valid relay parameter is either received by a sink and turned into a relay or remains valid.
\end{lemma}
\begin{proof}
 Note that a relay parameter is created with a unique key (which satisfies Property~2) and with $ID = r'.ID$ where $r'$ is the relay from which the relay parameter is created.
 Thus, if $r'$ is valid, Property~3 holds.
 Since we assume the relay parameter is sent via a valid relay, Property~1 initially holds also.
 For Property~4 to Property~6, also check the way a relay parameter is created and for Property~5 the fact that the relay $r''$ via which the relay parameter is sent is valid.
 For Property~7 to Property~11, note that since $key$ is uniquely created, no such message or relay can exist at that point in time. 
 
 We now check that every relay parameter remains valid as long as it remains in the buffer it is in.
 Property~1 remains valid due to Lemma~\ref{lem:valid_relay_remains_valid}.
 The same holds for Property~3 and the fact that relay parameters are never changed.
 This argument also implies Property~4.
 Property~2 again follows from the fact that relay parameters are created with a unique key.
 Let us now assume that Property~5 becomes false.
 By Lemma~\ref{lem:valid_relay_remains_valid} and the fact that the $sinkRID$ is never changed for a valid relay, this can only happen because $(key, \bot, r'')$ is removed from $r'.In$.
 According to the pseudocode, this happens at the following occasions:
 \begin{itemize}
  \item In Line~\ref{line:rl:r_in_remove_message_received}, in which case $r'$ receives a message using key $key$, which contradicts Property~8.
  \item In Line~\ref{line:rs:r_in_remove_not_authorized}, in which case prior to the change there must have existed a $\notauthorized{m = ((key,inID,outID,level),action(parameters))}$ message with $key \in r''.out.Key$ and $|r.out.Key| = 1$. 
  Since $r''$ is valid, this contradicts Property~\ref{item:vr:non_sink:one_key_correct}) together with Property~\ref{item:vr:notauthorized} of Definition~\ref{def:valid_relay} for $r'$.
  \item In Line~\ref{line:rs:r_in_remove_probefail}, but this requires a \probefail{key, keySequence} message with $key_1 \in r''.out.Key$ contradicting Property~10.
  \item In Line~\ref{line:rl:timeout_r_in_remove1}, but that line is only executed if Property~\ref{item:vr:in_key_unique} of $r'$ is violated, which cannot be the case according to Property~3.
  \item In Line~\ref{line:rl:timeout_r_in_remove2}, but that line is only executed if $r''$ does not exist, which contradicts Property~5.
  \item When \textbf{delete} $r'$ is called: this does not occur due to Lemma~\ref{lem:delete_never_called} and the fact that $r'.In \ne \emptyset$ and the assumption of the lemma.
  \item In Line~\ref{line:rl:r_in_remove_outrelayclosed}, but this requires an \outrelayclosed{ID} with $r''.out.ID = ID$ to be received which contradicts to the fact that $r''$ is valid and Property~\ref{item:vr:outrelayclosed} of Definition~\ref{def:valid_relay}
 \end{itemize}
 For Property~6 note that the relay references inside a message are never changed.
 For Property~7 note that this could only become violated if $r'''$ is created from a relay parameter with key $key$.
 This, however, would contradict Property~2.
 For Property~8 note that whenever the key of a message is set, it is set to a key from $s.out.Key$ for a relay $s$, which cannot be $key$ due to Property~7.
 For Property~9 note that whenever an existing \probefail{key, keySequence} message is changed, only $keySequence$ is truncated from the end.
 Thus, the only occasion at which Property~9 might become false is when a \probefail{key, keySequence} message is created.
 This only happens at a sink $s$ upon receipt of a \probe{controlKeys, keySequence} message with $key \in controlKeys$.
 However, this requires this \probe{} message to having been in $r'''.Buf$ for some relay $r'''$ with $r'''.out.ID = s.ID$, which contradicts Property~10 since $s.out.ID = \bot$. 
 For Property~10 check that every \probe{controlKeys, keySequence} message created does not violate this property due to the way these messages are created and the fact that for $key \in controlKeys$ they are sent via $r''$ only and $r''$ is valid.
 Thus, assume a \probe{} message $p$ is delivered from a buffer $r_i.Buf$ to the relay layer with \RID $RID(r_i.out.ID)$.
 There, for the relay $r_{i+1}$ with $r_{i+1}.ID = r_i.out.ID$, a key $key \in r_{i+1}.out.Key$ is appended to $keySequence$.
 If there is a message containing a relay parameter with a key contained in $keySequence$ (i.e., $m$ according to Property~2), this key is removed from $keySequence$.
 So either $p$ is no longer as defined in Property~10 or otherwise it fulfills the requirements of Property~10 still if they held before.
 For Property~11 check both cases in which a \inrelayclosed{Keys, RID, ID} message with $key \in Keys$ is sent.
 The first is in Line~\ref{line:rl:timeout:inrelayclosed} which requires a relay $r'''$ with $key \in r'''.out.Key$ contradicting Property~7.
 The second is in Line~\ref{line:rl:ping:inrelayclosed} which requires a \ping{ID, level, sinkRID, key} message which due to Property~5 of a valid relay parameter and Property~\ref{item:vr:ping} and Property~\ref{item:vr:in_key_unique} of a valid relay does not exist.

 Next assume the message containing the relay reference is transmitted from $r.Buf$ to the process whose relay layer has the \RID contained in $r.out.ID$ and let $s$ be the relay with $s.ID = r.out.ID$.
 There are two options now: 
 If $s$ is not a sink, i.e., $s.out.ID \ne \bot$, the message is put into $s.Buf$, and only Property~1 and Property~6 can have become invalid due to the transmission. 
 However, they still hold due to Property~\ref{item:vr:non_sink:nextValid}) and Property~\ref{item:vr:non_sink:level_and_rid_correct}) of a valid relay.
 If $s$ is a sink, due to Property~6, the corresponding message is not discarded by the protocol executed when the message is received.
 Thus according to the pseudocode and the fact that Property~7 holds, the relay parameter will be turned into a relay.
\end{proof}

\begin{lemma}\label{lem:valid_relay_parameter_yields_valid_relay}
 If the application does not delete a relay $r'$ if $r'.In \ne \emptyset$, each relay created from a valid relay parameter is a valid relay.
\end{lemma}
\begin{proof}
 Assume a relay $r$ is created from a valid relay parameter $(key, ID,level,\allowbreak sinkRID)$.
 This happens when the message $m$ containing the valid relay parameter is received by a relay layer such that the relay with ID $outID$ is a sink. 
 Note $r$ is created with $r.state = alive$, a globally unique $r.ID$ a pair $r.out$ and empty $r.In$.
 Thus, Property~\ref{item:vr:state_alive} to Property~\ref{item:vr:newprop} are statisfied.
 
 Furthermore $r$ is created such that $r.out.ID = ID$ and by Property~3 of a valid relay parameter, the relay $r'$ with $r'.ID = r.out.ID$ is valid, i.e., Property~\ref{item:vr:non_sink:outID}) and Property~\ref{item:vr:non_sink:nextValid}) hold true.
 
 According to Property~4 of a valid relay parameter, Property~\ref{item:vr:non_sink:level_and_rid_correct}) of a valid relay holds as well.
  
 For Property~\ref{item:vr:non_sink:one_key_correct}) of a valid relay, first note that before the relay parameter was received, the corresponding message was stored in $r'''.Buf$ for a relay $r'''$ before and $r'''.out.ID = outID$ (if the message was not delivered this way, in which case it must have been in the system initially, then the relay is not created from a valid relay parameter).
 By Property~1 of a valid relay parameter, $r'''$ was valid, thus $r'''.sinkRID = RID(outID)$ by Property~\ref{item:vr:non_sink:level_and_rid_correct}) of a valid relay.
 Thus, according to Property~5 of a valid relay parameter, $(key,\bot,r'') \in r'$ for some valid relay $r''$ and $r''.sinkRID = outID$.
 Since upon creation of the relay $r$, the message $((key,r.ID,r.out.ID), \probe{\{\}, key})$ is put into $r.Buf$, and $key$ was put into $r.out.Key$, and because of Property~11, Property~\ref{item:vr:non_sink:one_key_correct}) of a valid relay holds for $r$.
 
 For Property~\ref{item:vr:non_sink:out_key_unique}) check Line~\ref{line:rl:receipt_no_double_keys}. 
 
 Property~\ref{item:vr:non_sink:inrelayclosed}) follows from Property~11 of a valid relay parameter.
\end{proof}

\begin{lemma}\label{lem:sink_becomes_valid}
 If the application does not delete a relay $r'$ if $r'.In \ne \emptyset$, every alive relay $r$ with $r.out.ID = \bot$ becomes a valid relay unless it is deleted by the application.
\end{lemma}
\begin{proof}
 First of all notice that $r$ does not get deleted unless by the application according to Lemma~\ref{lem:delete_never_called}.
 Note that $r.state = alive$ by assumption (yielding Property~\ref{item:vr:state_alive}), $r.ID$ is globally unique by the assumption that there are no corrupted IDs in the system (yielding Property~\ref{item:vr:rid_unique}), Property~\ref{item:vr:rout_key} is satisfied by \timeout, and Property~\ref{item:vr:rin_stores_triples} and Property~\ref{item:vr:in_key_unique} are satisfied by Line~\ref{line:rl:timeout_r_in_remove1} and Line~\ref{line:rl:timeout_r_in_remove2} in Listing~\ref{lst:rl:timeout}, with one exception: 
 For those triples $(key,\bot,r'')$ that violate Property~\ref{item:vr:rin_stores_triples} in that $r''$ is not valid, note that all these triples will eventually be removed, either because they are replaced by the other kind of triples or becase the probing failed.
 
 Notice also that by Line~\ref{line:rl:timeout_level_set} and Line~\ref{line:rl:timeout_key_empty}, Property~\ref{item:vr:sink} of a valid relay will be fulfilled after the next execution of \timeout.
 
  For Property~\ref{item:vr:ping} check that any \ping{} message with first parameter $r.ID$ is sent only by $RL(r)$ as \ping{r.ID, r.level, r.sinkRID, key} and only such that $(key, RID, \bot) \in r.In$ for some $RID$ (this happens in Line~\ref{line:rl:ping_send}).
 Further note that as soon as Property~\ref{item:vr:sink}) holds, it will hold forever, i.e., $r.level$ and $r.sinkRID$ will never change.
 Thus as soon as after this point in time all existing \ping{} messages with first parameter $r.ID$ have been received, Property~\ref{item:vr:ping} holds and holds forever.

  For Property~\ref{item:vr:outrelayclosed} check that unless $r$ is deleted, an \outrelayclosed{r.ID} message could only be created in Line~\ref{line:rl:messagereceipt_outrelayclosed} by $RL(r)$ (due to the preceding ``if'') if $r$ is dead.
 Thus, as soon as all \outrelayclosed{r.ID} messages initially in the system are received, Property~\ref{item:vr:outrelayclosed} holds forever.
 
 To see that Property~\ref{item:vr:notauthorized} will eventually become true check in the pseudocode that a \notauthorized{m = ((key,inID,r.ID,level), action(parameters))} message is only sent in Line~\ref{line:rl:messagereceipt_notauthorized} and only by $RL(r)$.
 For this line to be executed, however, there must be no $(key,RID, \bot) \in r.In$ with $RID = RID(inID)$, and no $(key, \bot, r') \in r.In$
such that $r'.sinkRID = RID(inID)$.
 Thus such a message violating Property~\ref{item:vr:notauthorized} is never sent out and as soon as all of these message initially in the system have been received, Property~\ref{item:vr:notauthorized} will hold forever.
 
 For Property~\ref{item:vr:newprop}, we first note that every \probe{} message will not be forwarded infinitely often.
 Note that every message received by an arbitrary relay $r''$ is only forwarded to the relay with ID $r''.out.ID$ and that the value of $r''.out.ID$ cannot be changed for an existing relay.
 Thus, if a \probe{} message is forwarded infinitely often, then there must be a cycle $(r_0, \dots, r_k, r_{k+1} = r_0)$ of relays that this message traverses infinitely often.
 However, for some $l \in \{0, \dots, k\}$, $r_l.level < r_{l+1 mod k}.level + 1$ must hold, and any receipt of a \ping{} message by $RL(r_{l+1 mod k})$ can only increase $r_{l+1 mod k}.level$.
 Thus, at some point in time the \ping{ID, level, sinkRID, key} message sent during \timeout to $RL(r_l)$ with $ID = r_l.out.ID$ and $key \in r_l.out.Key$ will cause $r_l$ to be deleted and the cycle will be broken, yielding a contradiction.
 Second, not that every \probefail{key, keySequence} message cannot be forwarded infinitely often, as $keySequence$ is strictly decreasing during every forwarding.
 Thus, eventually all \probe{} and \probefail{} message initially in the system will be gone.
 The proof that no messages contradicting Property~\ref{item:vr:newprop} are created is analogous to the corresponding part in the proof of Lemma~\ref{lem:valid_relay_remains_valid}.
  
 All in all, as soon as \timeout is executed on $r$, and all initial messages have vanished, $r$ will be valid.
\end{proof}

For convenience, in the following we say a message $m$ has a target $r$ if $m \in r'.Buf$ for some relay $r'$ such that $r'.out.ID = r.ID$.

\begin{lemma}\label{lem:alive_non_sink_becomes_valid_or_gets_deleted}
 If the application does not delete a relay $r'$ if $r'.In \ne \emptyset$, for every alive relay $r$ that does not get merged with another relay or deleted by the application and such that with $r.out.ID = r'.ID$ for a valid relay $r'$, the following holds:
 If $r$ fulfills Property~\ref{item:vr:non_sink} 
 of a valid relay, $r$ will become a valid relay.
 Otherwise, $r$ will become a valid relay or become deleted in finite time.
\end{lemma}
\begin{proof}
 When $r$ gets merged or deleted by the application, we are done.
 Thus, in the following we assume that $r$ is not merged or deleted in finite time.

 Property~\ref{item:vr:state_alive} through Property~\ref{item:vr:in_key_unique} and Property~\ref{item:vr:outrelayclosed} through Property~\ref{item:vr:newprop} will be fulfilled after finite time for similar arguments as in Lemma~\ref{lem:sink_becomes_valid}.
 If Property~\ref{item:vr:non_sink} holds, Property~\ref{item:vr:ping} will also follow for similar reasons (here, Property~\ref{item:vr:non_sink:nextValid}) and Property~\ref{item:vr:non_sink:level_and_rid_correct}) imply that $r.level$ and $r.sinkRID$ will not change any more).
 Thus we are done in this case.

 In the following assume that $r$ does not get deleted in finite time at all (because otherwise we are done).
 Note that Property~\ref{item:vr:non_sink:outID}) and Property~\ref{item:vr:non_sink:nextValid}) holds by assumption. 
 We will now show that at some point in time either Property~\ref{item:vr:non_sink:out_key_unique}) will hold forever or $r$ will be deleted.
 After that we show that the same holds for Property~\ref{item:vr:non_sink:inrelayclosed}) and Property~\ref{item:vr:non_sink:one_key_correct}) and also for Property~\ref{item:vr:non_sink:level_and_rid_correct}).
 
   For Property~\ref{item:vr:non_sink:out_key_unique}) check that in case the property is violated, Line~\ref{line:rl:timeout_double_used_key_delete} of the \timeout action makes sure it becomes satisfied (either by removing $key$ from $r.out.Key$ or from $r'''.out.Key$).
 Furthermore, note that the only occasion at which a key $key$ is added to $r'''.out.Key$ for a relay $r'''$ owned by the same process as $r$ is in Line~\ref{line:rl:key_added}, in which due to Line~\ref{line:rl:receipt_no_double_keys} $key \notin r.out.Key$ holds.
 Thus, once Property~\ref{item:vr:non_sink:out_key_unique}) holds, it will hold forever.

  For Property~\ref{item:vr:non_sink:inrelayclosed}), note that according to the pseudocode, only $RL(r)$ could send such an \inrelayclosed{Keys, RID(r), r.out.ID} message with $key \in r.out.Key$.
  However, as long as $r$ is alive, according to the pseudocode, a message \inrelayclosed{Keys, RID(r), r.out.ID} with $key \in Keys$ can only be sent if $RL(r)$ receives a \ping{r.out.ID, level, sinkRID, key} and $key \notin r.out.Key$, yielding a contradiction.
  Thus, no such message is ever sent.
  Thus as soon as all \inrelayclosed{} messages initially in the system have vanished, no one contradicting Property~\ref{item:vr:non_sink:inrelayclosed}) will be created.

 Now assume Property~\ref{item:vr:non_sink:one_key_correct}) is violated.
 First of all, note that no $key$ is added to $r.out.Key$ because this only happens during the creation of $r$.
 Furthermore, note that every $((key,r.ID,r.out.ID), \probe{\{\}, key}) \in r.Buf$ will eventually be delivered to $r'$ causing $RL(r')$ to replace $(key,\bot,r'')$ by $(key,RID(r),\bot)$ in $r'.In$.
 Thus, eventually for every $key \in r.out.Key$, either $key$ is not the first parameter of a triple in $r'.In$ at all or $(key,RID,\bot) \in r'.In$.
 Observe that $r$ during \timeout will eventually send \probe{} messages with target $r'$ for every $key \in r.out.Key$.
 If for every key $key$ of these keys $(key, RID, \bot) \notin r'.In$ or $RID \ne RID(r)$, and $(key, \bot, r'') \notin r'.In$ or $r''.sinkRID \ne RID(r)$, according to the pseudocode in Listing~\ref{lst:rl:ch_added}, the relay layer of $r$ will receive \notauthorized{} messages for all of these messages causing it to remove the keys from $r.out.Key$ (note the little subtlety that $r.level$ might have changed in the meanwhile in which case the code is not executed, but this can happen only a finite number of times since $r.level$ is only decreased, see the pseudocode in Listing~\ref{lst:rl:ping}).
 Once the last key has been removed, $r$ will be deleted.
 Thus Property~\ref{item:vr:non_sink:one_key_correct}) holds eventually.
 Note that Property~\ref{item:vr:non_sink:one_key_correct}) holds forever for two reasons:
 First of all, a key is only removed from $r.out.Key$ if $RL(r)$ receives a $\notauthorized{m = ((key,inID,outID,level), action(parameters))}$ message with $r.ID = inID$, $r.out.ID = outID \ne \bot$, $r.level = level$ and $key \in r.out.Key$.
 Due to Property~\ref{item:vr:notauthorized} for $r'$, this cannot happen for at least one $key \in r.out.Key$.
 Second, a valid relay $r'$ only removes a $(key,\bot,r'')$ from $r'.In$ if it replaces them by $(key, RID, \bot)$ and it only removes a $(key, RID, \bot)$ from $r'.In$ if it receives a \inrelayclosed{Keys, RID, r'.ID} message with $key \in Keys$.
 This cannot happen since we already proved that there will be no message contradicting Property~\ref{item:vr:non_sink:inrelayclosed}).
 
 For Property~\ref{item:vr:non_sink:level_and_rid_correct}) assume that Property~\ref{item:vr:non_sink:one_key_correct}) already holds (which will be the case as we have just proven) and let $key \in r.out.Key$ such that (a) $(key, RID, \bot) \in r'.In$ and $RID = RID(r)$, or (b) $(key, \bot, r'') \in r'.In$ and $r''.sinkRID = RID(r)$ and $((key,r.ID,r.out.ID), \probe{\{\}, key}) \in r.Buf$.
 We consider both cases individually.
 In Case~(a), during \timeout, $RL(r')$ will send a \ping{r'.ID, r'.level, r'.sink, key} message with target $r$.
 Upon receipt of this message, $RL(r)$ will either delete $r$ or update the values of $r$ such that Property~\ref{item:vr:non_sink:level_and_rid_correct}) is fulfilled.
 In Case~(b), the $((key,r.ID,r.out.ID), \probe{\{\}, (key)})$ message in $r.Buf$ will eventually be delivered to $RL(r')$ in which case according to Line~\ref{line:rl:r_in_remove_message_received}f. in Listing~\ref{lst:rl:ch_added}, $(key, \bot, r'')$ in $r'.In$ will be replaced by $(key, RID(r.ID), \bot)$.
 After this we are in Case~(a).
 Note that since Property~\ref{item:vr:ping} holds for $r'$, and $r.level$ and $r.sinkRID$ are only changed due to the receipt of a \ping{} message, Property~\ref{item:vr:non_sink:level_and_rid_correct}) will hold forever.

 All in all, we obtain the claim of the lemma.
\end{proof}

\begin{lemma}\label{lem:merge_from_valid_creates_valid}
 For every call of $merge(R)$ such that there is a valid relay $r \in R$, either merge does nothing or the resulting relay $r'$ is a valid relay.
\end{lemma}
\begin{proof}
 Check that if $merge(R)$ does not do nothing then for all relays $r' \in R$, $r'.state=alive$, $r'.out.ID = r.out.ID$, $r'.level = r.level$, $r'.sinkRID = r.sinkRID$, and $r.In=\{ \}$.
 Furthermore, for the new relay $r''$ created it holds that $r''.ID$ is a new ID, $r''.state=alive$, and $r''.out=(Key,ID)$ with $r.Key \subseteq Key$ $r''.level=r'.level$, $r''.In=\{ \}$, and $r''.Buf=\bigcup_{r \in R} r.Buf$.
 Thus Property~\ref{item:vr:state_alive} through Property~\ref{item:vr:notauthorized} follow immediately from the fact that $r$ was valid an the fact that $r''.In = \emptyset$.
 PropertyR10 also follows from the fact that $r$ was valid and that all relays $r' \in R$ are deleted during merge.
 
\end{proof}

\begin{proofof}{Theorem~\ref{thm:rl:messages_sent_via_valid_relay_will_be_received_by_sink}}
 Follows from Property~\ref{item:vr:rout_key} and Property~\ref{item:vr:rin_stores_triples} of a valid relay and the fact that according to the pseudocode the message is always forwarded from $r$ to $r'$ (as defined in Definition~\ref{def:valid_relay} until it is received by a process with $r.out.ID = \bot$.
\end{proofof}

\begin{proofof}{Theorem~\ref{thm:rl:valid_relay_initially_will_remain_valid}}
 Note that every sink relay created (via \textbf{new Relay}) becomes a valid relay.
 Furthermore, according to Lemma~\ref{lem:valid_relay_remains_valid}, every valid relay remains a valid relay.
 Additionally, every relay reference created becomes a valid relay reference according to Lemma~\ref{lem:valid_relay_parameter_created_is_valid} (note that it can only be created from a valid relay and only sent via a valid relay).
 Furthermore, by that lemma, every relay reference remains a valid reference.
 Besides, by Lemma~\ref{lem:valid_relay_parameter_yields_valid_relay} every relay created by a message receipt will be a valid relay.
 Last, by Lemma~\ref{lem:merge_from_valid_creates_valid} every relay created by a merge will be a valid relay.
 Since there are no other occasions at which relays are created, the claim follows.
\end{proofof}

\begin{proofof}{Theorem~\ref{thm:rl:fair_delete_no_indirect_forwarding}}
Note that every sink relay created becomes a valid relay and that by Lemma~\ref{lem:sink_becomes_valid}, every sink relay will eventually become valid.
By Lemma~\ref{lem:valid_relay_remains_valid}, they will remain valid forever.
Thus starting from some state $S$, only valid sink relays will exist.

Note that after $S$, every relay parameter of a sink relay created is a valid relay parameter by Lemma~\ref{lem:valid_relay_parameter_created_is_valid}.
Thus as soon as all relay parameters still in the system in state $S$ have been received or deleted, say at state $S'$, all relay parameters created from sink relays will be valid.
Thus, every relay $r$ created such that $r.out.ID = r'.ID$ for a sink relay $r'$ is a valid relay.
Furthermore, every existing relay $r$ such that $r.out.ID = r'.ID$ for a sink relay $r'$ will become valid or get deleted by Lemma~\ref{lem:alive_non_sink_becomes_valid_or_gets_deleted}.
Thus at some state $S''$, all direct relays are valid and will remain valid by Lemma~\ref{lem:valid_relay_remains_valid}.

The claim then follows from the assumption of the lemma and Lemma~\ref{lem:valid_relay_parameter_created_is_valid}, Lemma~\ref{lem:valid_relay_parameter_yields_valid_relay} and Lemma~\ref{lem:merge_from_valid_creates_valid}.
\end{proofof}

\begin{proofof}{Theorem~\ref{thm:rl_of_stopped_process_will_eventually_be_destroyed}}
 Recall that when a process $v$ executes \textbf{stop}, $RL(v)$ immediately deletes all sink relays and periodically deletes all relays $r$ with $r.In = \emptyset$ and $r.Buf = \emptyset$.
 
 First of all, assume for contradiction that on an inactive process there is a relay $r$ such that $r.In \ne \emptyset$ forever.
 Since $r$ is not a sink relay, $r.out.ID \ne \emptyset$.
 
 We first consider all $(key,\bot, r'') \in r.In$.
 Consider the sequence of relays $(r_1 = r'', r_2, \dots, r_k)$, such that $r_{i+1}.ID = r_i.out.ID$ for all $1 \le i < k$ and $r_k.out.ID = \bot$ (note that this sequence must be finite or at some point in time upon receipt of a \ping{} there will be a mismatch causing a relay on that sequence to be deleted).
 If all of the relays in this sequence are never deleted, they will become valid for similar reasons as in the proof of Lemma~\ref{lem:alive_non_sink_becomes_valid_or_gets_deleted}.
 Otherwise, all relays of the sequence will eventually be deleted for similar arguments as in this proof.
 Thus, eventually either $r''$ will be deleted (in which case $(key,\bot, r'')$ will be removed from $r.In$) or $r$ becomes valid.
 In this case, however, according to Property~\ref{item:vr:newprop} at some point in time, $r_k$ will receive a message containing a relay parameter with key $key$ upon which it will create a \probe{} message to $r$ with that key, or it already has.
 As soon as this message is received, $(key,\bot, r'')$ will be replaced by $(key, RID, \bot$ in $r.In$.
 Note that no new $(key,\bot, r'')$ are added to $r.In$ because the process owning $r$ is dead.
 
 Now we consider the $(key,RID,\bot) \in r.In$.
 Due to the pseudocode of \timeout in Listing~\ref{lst:rl:timeout}, $r.level \geq 1$ will eventually hold.
 Thus, $r$ will for every $(key,RID,\bot) \in r.In$ eventually send a \ping{r.ID, r.level, r.sinkRID, key} message to the relay layer with \RID $RID$. 
 Upon receipt of this message, either a message \inrelayclosed{\{key\}, \allowbreak RID, r.ID} will be sent to $RL(r)$ causing $r$ to remove $(key,RID,\bot)$ from $r.In$, or for a relay $r''$ with $key \in r''.out.Key$, $r.level \geq 2$ afterwards or that one is deleted.
 In both cases it will eventually be deleted (in the first case according to the assumption).
 Thus at some point in time the \ping{r.ID, r.level, r.sinkRID, key} message will be returned by an \inrelayclosed{\{key\}, RID, r.ID} causing $r$ to remove the entry from $r.In$.
 Note that a $(key,RID,\bot)$ is only added to $r.In$ if before there was a corresponding $(key,\bot, r'') \in r.In$.
 Thus, at some point in time all element in $r.In$ will have vanished.
 
 Now assume that $r.In = \emptyset$ but $r.Buf \ne \emptyset$.
 Note that in this case no message is added to $r.Buf$ according to the pseudocode.
 Thus, eventually $r.In = \emptyset$ and $r.Buf = \emptyset$, after which $r$ will be deleted.
 
 Once all relays of a relay layer have been deleted, it will shut down during \timeout.
\end{proofof}

\bibliographystyle{plain}
\bibliography{bibliography.bib}

\end{document}